\documentclass[sigconf,nonacm]{acmart}

\usepackage{amsthm}
\usepackage[binary-units=true]{siunitx}
\usepackage{amsmath}

\newcommand*{\symDefine}[2]{\newcommand*{#1}{#2}}

\symDefine{\symMVP}{\text{MVP}}
\symDefine{\symUpdateEvent}{A}
\symDefine{\symMasked}{a}
\symDefine{\symThreshold}{a}
\symDefine{\symBase}{b}
\symDefine{\symBiasCorrectionFactorConstant}{c}
\symDefine{\symSomeConstant}{c}
\symDefine{\symNumExtraBits}{d}
\symDefine{\symFunc}{f}
\symDefine{\symOtherFunc}{g}
\symDefine{\symRegMartingale}{h}
\symDefine{\symHashBit}{h}
\symDefine{\symRegAddr}{i}
\symDefine{\symIndexJ}{j}
\symDefine{\symIndexL}{l}
\symDefine{\symUpdateVal}{k}
\symDefine{\symIndexBit}{l}
\symDefine{\symNumReg}{m}
\symDefine{\symCardinality}{n}
\symDefine{\symCardinalityMax}{\symCardinality_\text{\normalfont max}}
\symDefine{\symCardinalityEstimatorML}{\symCardinalityEstimator_\text{\normalfont ML}}
\symDefine{\symCardinalityEstimator}{\hat{\symCardinality}}
\symDefine{\symCardinalityEstimatorMartingale}{\symCardinalityEstimator_\text{\normalfont martingale}}
\symDefine{\symPrecision}{p}
\symDefine{\symProbability}{\Pr}
\symDefine{\symBitsForMax}{q}
\symDefine{\symRegister}{r}
\symDefine{\symHashTokenBit}{s}
\symDefine{\symShift}{s}
\symDefine{\symExtraBits}{t}
\symDefine{\symTimeIndex}{t}
\symDefine{\symTokenSet}{T}
\symDefine{\symMaxUpdateVal}{u}
\symDefine{\symMaxUpdateValMin}{{\symMaxUpdateVal_\text{\normalfont min}}}
\symDefine{\symMaxUpdateValMax}{{\symMaxUpdateVal_\text{\normalfont max}}}
\symDefine{\symHashTokenParameter}{v}
\symDefine{\symNumBits}{v}
\symDefine{\symToken}{w}
\symDefine{\symX}{x}
\symDefine{\symXOld}{\symX_\text{\normalfont old}}
\symDefine{\symXZero}{\hat{\symX}}
\symDefine{\symY}{y}
\symDefine{\symZ}{z}

\symDefine{\symLikelihood}{\mathcal{L}}
\symDefine{\symShannon}{\mathcal{H}}
\symDefine{\symComplexity}{\mathcal{O}}

\symDefine{\symExpFunc}{\phi}
\symDefine{\symAlphaContribFunc}{\omega}
\symDefine{\symLikelihoodFuncExponentOne}{\alpha}
\symDefine{\symLikelihoodFuncExponentTwo}{\beta}
\symDefine{\symLikelihoodFuncExponentTwoSum}{\sigma_0}
\symDefine{\symLikelihoodFuncExponentTwoPow}{\sigma_1}
\symDefine{\symGammaFunc}{\Gamma}
\symDefine{\symDensity}{\rho}
\symDefine{\symDensityUpdate}{\symDensity_\text{\normalfont update}}
\symDefine{\symDensityRegister}{\symDensity_\text{\normalfont reg}}
\symDefine{\symDensityToken}{\symDensity_\text{\normalfont token}}
\symDefine{\symZetaFunc}{\zeta}
\symDefine{\symStateChangeProbability}{\mu}
\symDefine{\symDelta}{\Delta}
\symDefine{\symSumA}{\varphi}
\symDefine{\symSumB}{\psi}
\symDefine{\symFactorOne}{\lambda}
\symDefine{\symFactorTwo}{\eta}

\DeclareMathOperator*{\symVariance}{Var}

\DeclareMathOperator*{\symNLZ}{nlz}
\DeclareMathOperator*{\symExpMinusOne}{expm1}
\DeclareMathOperator*{\symLogPlusOne}{log1p}

\symDefine{\symBitwiseAnd}{\mathrel{\&}}
\symDefine{\symBitwiseOr}{\mathrel{|}}

\usepackage{xifthen}
\usepackage[nolist,nohyperlinks]{acronym}
\usepackage[ruled,vlined]{algorithm2e}
\usepackage{mathtools}
\usepackage[capitalise,noabbrev]{cleveref}
\usepackage{enumitem}
\usepackage{multirow}
\usepackage{makecell}
\usepackage{threeparttable}
\usepackage{minibox}
\usepackage{booktabs} % For formal tables

\crefname{algocf}{alg.}{algs.}
\Crefname{algocf}{Algorithm}{Algorithms}

\SetKwProg{Fn}{function}{}{}
\SetKwFunction{FuncMergeReg}{MergeRegister}
\SetKwFunction{FuncDownReg}{DownsizeRegister}
\SetKwComment{Comment}{$\triangleright$\ }{}

\SetCommentSty{mycommfont}

\newcommand{\myAlg}[2][]{
 \ifthenelse{\isempty{#1}}%
 {\begin{figure}}% if #1 is empty
 {\begin{figure}[#1]}% if #1 is not empty
 \begingroup % trick algorithm2e into thinking we're in one column mode
 \csname @twocolumnfalse\endcsname
 \noindent
 \resizebox{\columnwidth}{!}{%
 \begin{minipage}{1.2\columnwidth}
 \begin{algorithm}[H]
 \DontPrintSemicolon
 {#2}
 \end{algorithm}
 \end{minipage}%
 }% <------------- end of \resizebox
 \endgroup
 \end{figure}
}

\renewcommand{\eqref}[1]{\textup{(\ref{#1})}}

\SetKwFor{Loop}{loop}{}{end loop}
\SetKw{Break}{break}
\SetKw{Or}{or}
\SetKw{Not}{not}

\begin{acronym}
 \acro{HLL}{HyperLogLog}
 \acro{EHLL}{ExtendedHyperLogLog}
 \acro{ULL}{UltraLogLog}
 \acro{HLLL}{HyperLogLogLog}
 \acro{ELL}{ExaLogLog}
 \acro{PCSA}{probabilistic counting with stochastic averaging}
 \acro{ML}{maximum likelihood}
 \acro{RMSE}{root-mean-square error}
 \acro{NLZ}{number of leading zeros}
 \acro{CPC}{compressed probability counting}
 \acro{HIP}{historic inverse probability}
 \acro{MVP}{memory-variance product}
 \acro{PMF}{probability mass function}
 \acro{FISH}{Fisher-Shannon}
\acro{CPU}{central processing unit}
\end{acronym}

\allowdisplaybreaks
\DontPrintSemicolon
\begin{document}
\title[ExaLogLog: Space-Efficient and Practical Approximate Distinct Counting up to the Exa-Scale]{ExaLogLog: Space-Efficient and Practical Approximate\\Distinct Counting up to the Exa-Scale}

\author{Otmar Ertl}
\affiliation{
	\institution{Dynatrace Research}
	\city{Linz}
	\country{Austria}
}
\email{otmar.ertl@dynatrace.com}

\begin{abstract}
	This work introduces ExaLogLog, a new data structure for approximate distinct counting, which has the same practical properties as the popular HyperLogLog algorithm. It is commutative, idempotent, mergeable, reducible, has a constant-time insert operation, and supports distinct counts up to the exa-scale. At the same time, as theoretically derived and experimentally verified, it requires \qty{43}{\percent} less space to achieve the same estimation error.
\end{abstract}

\maketitle

\section{Introduction}

Exact counting of distinct elements in a data set or a data stream is known to take linear space \cite{Alon1999}. However, the space requirements can be significantly reduced, if approximate results are sufficient. \ac{HLL} \cite{Flajolet2007} with improved small-range estimation \cite{Heule2013, Qin2016, Zhao2016, Ertl2017, Ting2019} has become the standard algorithm for approximate distinct counting. It achieves a relative standard error of $1.04/\sqrt{\symNumReg}$ up to distinct counts in the order of $2^{64}\approx1.8\cdot10^{19}$ using only $6\symNumReg$ bits \cite{Heule2013}.
Therefore, the query languages of many data stores (see e.g. documentation of Timescale, Redis, Oracle Database, Snowflake, Microsoft SQL Server, Google BigQuery, Vertica, Elasticsearch, Aerospike, Amazon Redshift, KeyDB, or DuckDB) offer special commands for approximate distinct counting that are usually based on \ac{HLL}. Query optimization \cite{Freitag2019, Pavlopoulou2022}, caching \cite{Wires2014}, graph analysis \cite{Boldi2011, Priest2018}, attack detection \cite{Chabchoub2014, Clemens2023}, network volume estimation \cite{Basat2018}, or metagenomics \cite{Baker2019, Marcais2019, Elworth2020, Breitwieser2018} are further applications of \ac{HLL}.

\ac{HLL} is actually very simple as exemplified in \Cref{alg:insertion_hll}. It typically consists of a densely packed array of 6-bit registers $\symRegister_0, \symRegister_1, \ldots, \symRegister_{\symNumReg-1}$ where the number of registers $\symNumReg$ is a power of 2, $\symNumReg = 2^\symPrecision$ \cite{Heule2013}. The choice of the precision parameter $\symPrecision$ allows trading space for better estimation accuracy. Adding an element requires calculating a 64-bit hash value. $\symPrecision$ bits are used to choose a register for the update. The \ac{NLZ} of the remaining $64-\symPrecision$ bits are interpreted as a geometrically distributed update value $\geq 1$ with success probability $\frac{1}{2}$, that is used to update the selected register. A register always holds the maximum of all its previous update values. Estimating the distinct count from the register values is more challenging, but can also be implemented using a few lines of code \cite{Flajolet2007, Ertl2017}.
\myAlg[t]{
	\caption{Inserts an element with 64-bit hash value $\langle\symHashBit_{63} \symHashBit_{62} \ldots \symHashBit_{0}\rangle_2$ into a \acl*{HLL} consisting of $\symNumReg = 2^\symPrecision$ ($\symPrecision\geq 2$) 6-bit registers $\symRegister_0, \symRegister_1, \ldots, \symRegister_{\symNumReg-1}$ with initial values $\symRegister_\symRegAddr = 0$.}
	\label{alg:insertion_hll}
	$\symRegAddr\gets \langle\symHashBit_{63} \symHashBit_{62}\ldots\symHashBit_{64-\symPrecision}\rangle_2$\Comment*[r]{extract register index}
	$\symMasked\gets \langle\,\underbracket[0.5pt][1pt]{0\ldots 0}_{\scriptscriptstyle\symPrecision}\!\symHashBit_{63-\symPrecision}\symHashBit_{62-\symPrecision} \ldots\symHashBit_{0}\rangle_2$\Comment*[r]{mask register index bits}
	$\symUpdateVal \gets \symNLZ(\symMasked) -\symPrecision + 1$\Comment*[r]{\minibox[t]{update value $\symUpdateVal \in [1, 65-\symPrecision]\subseteq [1,63]$\\$\symNLZ$ returns the \acl*{NLZ}}}
	$\symRegister_\symRegAddr \gets \max(\symRegister_\symRegAddr, \symUpdateVal)$\Comment*[r]{update register}
}
\ac{HLL} owes its popularity to the following features allowing it to be used in distributed systems \cite{Ertl2024}:
\begin{description}[style=unboxed,leftmargin=0cm]
	\item[Speed:] Element insertion is a fast and allocation-free operation with a constant time complexity independent of the sketch size. In particular, given the hash value of the element, the update requires only a few CPU instructions.
	\item [Idempotency:] Further insertions of the same element will never change the state. This is actually a natural property every algorithm for distinct counting should support to prevent duplicates from changing the result.
	\item [Mergeability:] Partial results calculated over subsets can be easily merged to a final result. This is important when data is distributed, is processed in parallel, or needs to be aggregated.
	\item [Reproducibility:] The result does not depend on the processing order, which often cannot be guaranteed in practice anyway. Reproducibility is achieved by a commutative insert operation and a commutative and associative merge operation.
	\item [Reducibility:] The state can be reduced to a smaller state corresponding to a smaller precision parameter. The reduced state is identical to that obtained by direct recording with lower precision. This property allows adjusting the precision without affecting the mergeability with older records.
	\item [Estimation:] A fast and robust estimation algorithm ensures nearly unbiased estimates with a relative standard error bounded by a constant over the full range of practical distinct counts.
	\item [Simplicity:] The implementation requires only a few lines of code. The entire state can be stored in a single byte array of fixed length which makes serialization very fast and convenient. Furthermore, add and in-place merge operations do not require any extra memory allocations.
\end{description}

Until recently \ac{HLL} was the most space-efficient practical data structure having all these desired properties. Space efficiency can be measured in terms of the \ac{MVP} \cite{Pettie2021a}, which is the relative variance of the (unbiased) distinct-count estimate $\symCardinalityEstimator$ multiplied by the storage size in bits
\begin{equation}
	\label{equ:mvp_def}
	\symMVP := \symVariance(\symCardinalityEstimator/\symCardinality)\times(\text{storage size in bits}),
\end{equation}
where $\symCardinality$ is the true distinct count.
If the \ac{MVP} is asymptotically (for sufficiently large $\symCardinality$) a constant specific to the data structure, it can be used for comparison as it eliminates the general inverse dependence of the relative estimation error on the root of the storage size (see \Cref{fig:memory_over_error}).
Most \ac{HLL} implementations use 6-bit registers \cite{Heule2013} to support distinct counts beyond the billion range, resulting in a theoretical \ac{MVP} of 6.48 \cite{Pettie2021a}. A recent theoretical work conjectured a general lower bound of 1.98 for the \ac{MVP} of sketches supporting mergeability and reproducibility \cite{Pettie2021}, which shows the potential for improvement.

\begin{figure}
	\centering
	\includegraphics[width=\linewidth]{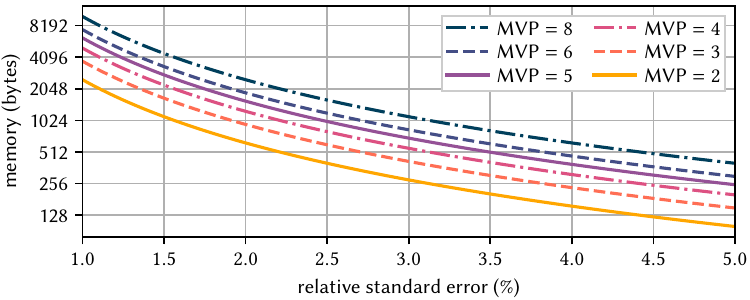}
	\caption{\boldmath The memory over the relative standard error for different \acfp*{MVP} following \eqref{equ:mvp_def}.}
	\label{fig:memory_over_error}
\end{figure}

\subsection{Related Work}
In the past, different approaches have been proposed to outperform the space efficiency of \ac{HLL}. However, many of them sacrificed at least one of the properties listed above \cite{Chen2011,Sedgewick2022,Xiao2020,Lu2023,Pettie2021a,Janson2024}. Therefore, in the following, we focus only on algorithms that support at least mergeability, idempotency, and reproducibility, which we consider to be essential properties for practical applications in distributed systems.

Lossless compression can significantly reduce the storage size of \ac{HLL} \cite{Durand2004,Scheuermann2007,Lang2017}. Since the compressed state prevents random access to registers as required for insertion, bulking is needed to realize at least amortized constant update times. The required buffer partially cancels out the memory savings. Recent techniques avoid buffering of insertions. The Apache DataSketches library \cite{ApacheDataSketches} provides an implementation using 4 bits per register to store the most frequent values relative to a global offset. Out of range values are kept separately in an associative array. Overall, this leads to a smaller \ac{MVP} but also to a more expensive insert operation. Its runtime is proportional to the memory size in the worst case, because all registers must be updated whenever the global offset is increased. \ac{HLLL} \cite{Karppa2022} takes this strategy to the extreme with 3-bit registers and achieves a space saving of around \qty{40}{\percent} at the expense of an insert operation which, except for very large numbers, has been reported to be on average more than an order of magnitude slower compared to \ac{HLL} \cite{Karppa2022}.

Compression can also be applied to other data structures. \Ac{PCSA} \cite{Flajolet1985} is a predecessor of \ac{HLL}, also known as FM-sketch. Although less space-efficient when uncompressed, its compressed state leads to a smaller \ac{MVP} than that of \ac{HLL} \cite{Scheuermann2007, Lang2017}. The \ac{CPC} sketch as part of the Apache DataSketches library \cite{ApacheDataSketches} uses this finding. The serialized compressed representation of the \ac{CPC} sketch achieves a \ac{MVP} of around 2.31 \cite{SketchesFeatureMatrix} that is already quite close to the conjectured lower bound of 1.98 \cite{Pettie2021}. However, this is achieved by a costly consolidation and compression step, as the in-memory representation is more than twice as large, since an amortized constant runtime of insertions can only be achieved by bulking.

Recently, data structures have been proposed that are more space-efficient than \ac{HLL} while not giving up constant-time insertions. \ac{EHLL} \cite{Ohayon2021} extends the \ac{HLL} registers from 6 to 7 bits to store not only the maximum update value, but also whether there was an update with a value smaller by one. This additional information can be used to obtain more accurate estimates. In particular, the \ac{MVP} is reduced by \qty{16}{\percent} to \num{5.43}. Inspired by this result,
we analyzed, if even more additional bits used to store the information about the occurrence of smaller update values could further improve the \ac{MVP} \cite{Ertl2024}. The result of our theoretical analysis led to \ac{ULL}, which uses two extra bits and achieves a \ac{MVP} of 4.63, an improvement of \qty{28}{\percent} over \ac{HLL}. Our theoretical analysis also showed that further space savings would be possible with even more additional bits if, at the same time, success probabilities smaller than $\frac{1}{2}$ were used for the geometric distribution of the update values. However, as the generation of such update values is more complicated, this approach has not been pursued so far.

Another recent data structure, SpikeSketch \cite{Du2023}, has chosen a different approach and combines update values following a geometric distribution with a success probability $\frac{3}{4}$ and a special lossy encoding scheme to improve space efficiency. It preserves mergeability and idempotency and also has a constant update time. Unfortunately, our experiments presented later could not confirm the claimed \ac{MVP} of 4.08 over the entire range of distinct counts. We found values that are significantly larger for distinct counts below \num{e4} caused by lossy compression and also the stepwise smoothing.

This work focuses on distinct counting of single, possibly distributed,
data flows. Data sketches for (non-distinct) counting \cite{Morris1978, Stanojevic2007, Punter2023} or distinct counting of multiple flows \cite{Xiao2015, Ting2019} were thus out of scope.

\begin{table}[t]
	\caption{Notations.}
	\label{tab:notation}
	\scriptsize
	\begin{tabular*}{\linewidth}{@{\extracolsep{\fill}} ll }
		\toprule
		Symbol
		&
		Description
		\\
		\midrule
		$\lfloor\ldots\rfloor$
		&
		floor function, e.g. $\lfloor 3.7\rfloor = 3$
		\\
		$\langle\ldots\rangle_2$
		&
		binary representation, e.g. $\langle 110\rangle_2 = 6$
		\\
		$\symBitwiseOr$
		&
		bitwise OR operation, e.g. $\langle1001\rangle_2 \symBitwiseOr \langle1010\rangle_2 = \langle1011\rangle_2$
		\\
		$\symBitwiseAnd$
		&
		bitwise AND operation, e.g. $\langle1001\rangle_2 \symBitwiseAnd \langle1010\rangle_2 = \langle1000\rangle_2$
		\\
		$\symNLZ$
		&
		\acl*{NLZ} if the argument is interpreted as unsigned 64-bit value,
		\\
		&
		e.g. $\symNLZ(\langle 10110\rangle_2)=59$
		\\
		$\symCardinality$
		&
		distinct count
		\\
		$\symCardinalityEstimator$
		&
		distinct count estimate
		\\
		$\symBase$
		&
		base, $\symBase > 1$, defines distribution of update values, compare \eqref{equ:geometric}
		\\
		$\symPrecision$
		&
		precision parameter
		\\
		$\symNumReg$
		&
		number of registers, $\symNumReg = 2^\symPrecision$
		\\
		$\symBitsForMax$
		&
		number of bits used for storing an update value
		\\
		$\symNumExtraBits$
		&
		number of additional register bits to indicate updates with smaller values
		\\
		$\symExtraBits$
		&
		parameter of approximated update value distribution, compare \eqref{equ:step_dist}
		\\
		$\symRegister_\symRegAddr$
		&
		value of $\symRegAddr$-th register, $0\leq \symRegAddr < \symNumReg$, $0\leq \symRegister_\symRegAddr \leq (65 - \symPrecision - \symExtraBits)2^{\symExtraBits+\symNumExtraBits}+2^{\symNumExtraBits}-1$
		\\
		$\symDensityUpdate$
		&
		\acl*{PMF} of update values, see \eqref{equ:geometric} and \eqref{equ:step_dist}
		\\
		$\symDensityRegister$
		&
		\acl*{PMF} of register values, see \Cref{sec:pmf_reg}
		\\
		$\symDensityToken$
		&
		\acl*{PMF} of hash tokens, see \eqref{equ:pmf_hash_token}
		\\
		$\symHashTokenParameter$
		&
		hash token parameter, hash token takes $\symHashTokenParameter+6$ bits, see \Cref{sec:sparse_mode}
		\\
		$\symLikelihood$
		&
		likelihood function, $\symLikelihood = \symLikelihood(\symCardinality\vert\symRegister_0\ldots \symRegister_{\symNumReg-1})$, see \Cref{sec:ml_estimation,sec:sparse_mode}
		\\
		$\symGammaFunc$
		&
		gamma function, $\symGammaFunc(\symX):=\int_0^\infty \symY^{\symX-1}e^{-\symY}d\symY$
		\\
		$\symZetaFunc$
		&
		Hurvitz zeta function, $\symZetaFunc(\symX, \symY) := \sum_{\symMaxUpdateVal=0}^\infty (\symMaxUpdateVal + \symY)^{-\symX}= \frac{1}{\symGammaFunc(\symX)}\int_0^\infty \frac{\symZ^{\symX-1}e^{-\symY\symZ}}{1-e^{-\symZ}}d\symZ$
		\\
		$\symStateChangeProbability$
		&
		probability that the next distinct element causes a state change, see \Cref{sec:martingale_estimation}
		\\
		$\symExpMinusOne$
		&
		$\symExpMinusOne(\symX):=e^\symX -1$, built-in function available in most standard libraries
		\\
		$\symLogPlusOne$
		&
		$\symLogPlusOne(\symX):=\ln(1 + \symX)$, built-in function available in most standard libraries
		\\
		\bottomrule
	\end{tabular*}
\end{table}

\subsection{Summary of Contributions}

We introduce \acf{ELL}, which is based on a recently proposed and theoretically analyzed data structure \cite{Ertl2024} that generalizes \ac{HLL}, \ac{EHLL}, \ac{ULL}, and \ac{PCSA}. However, we replace the geometric distribution of the update values by a distribution for which it is easier to map a 64-bit hash value to a corresponding random value. When optimally configured, \ac{ELL} achieves a \ac{MVP} of 3.67 as theoretically predicted and experimentally confirmed. Compared to \ac{HLL} with 6-bit registers \cite{Heule2013}, which is still the standard algorithm in most applications, \ac{ELL} supports the same operating range up to the exa-scale, but requires up to \qty{43}{\percent} less space. In contrast to \acf{CPC} \cite{ApacheDataSketches,Lang2017} or \acf{HLLL} \cite{Karppa2022}, the insert operation takes constant time independent of the sketch size and therefore also independent of the configured precision.

We will show that the new update value distribution leads to a simple \ac{ML} equation for the \ac{ELL} state, which has only a small number of terms regardless of the chosen precision. We propose a robust and efficient algorithm based on Newton’s method to solve this equation and, hence, to quickly find the distinct count estimate.

Moreover, we considered martingale estimation, which is the standard estimation algorithm, if the data is not distributed and merging is not needed \cite{Ting2014,Cohen2015}. For this case, the theory also predicts significant reductions of the \ac{MVP} by up to \qty{33}{\percent} compared to \ac{HLL}, which was again confirmed experimentally.

Finally, we present a strategy to realize a sparse mode for \ac{ELL} by transforming the 64-bit hash values to smaller hash tokens which can be collected and equivalently used for insertions. We also demonstrate how the distinct count can be directly estimated from those hash tokens using \ac{ML} estimation.

A reference implementation of \ac{ELL} written in Java can be found at \url{https://github.com/dynatrace-research/exaloglog-paper}. This repository also contains all the source code and instructions for reproducing the results and figures presented in this work. \Cref{tab:notation} summarizes the notations used in the following.

\section{Data Structure}
\label{sec:data_structure}

Recently, we have introduced and theoretically analyzed a data structure that can be seen as generalization of \ac{HLL}, \ac{EHLL} and \ac{PCSA}, and that finally led us to \ac{ULL}, another special case \cite{Ertl2024}.
The generalized data structure consists of $\symNumReg = 2^\symPrecision$ registers. The update operation uses $\symPrecision$ bits of a uniformly distributed hash value for selecting a register and another geometrically distributed hash value with \ac{PMF}
\begin{equation}
	\label{equ:geometric}
	\symDensityUpdate(\symUpdateVal) = (\symBase - 1)\symBase^{-\symUpdateVal} \qquad \symUpdateVal \geq 1, \symBase > 1
\end{equation}
for updating the selected register. Each register consists of $\symBitsForMax + \symNumExtraBits$ bits. The first $\symBitsForMax$ bits store the maximum $\symMaxUpdateVal$ of all update values the register was updated with, and the remaining $\symNumExtraBits$ bits indicate the occurrences of update values from the range $[\symMaxUpdateVal - \symNumExtraBits,\symMaxUpdateVal - 1]$.
By definition, this data structure is idempotent, as duplicate insertions will never modify the state. Since the state may only change for distinct insertions, it enables distinct count estimation.

\subsection{Previous Theoretical Results}
\label{sec:prev_results}
We have derived and presented various expressions for the \ac{MVP} of this generalized data structure \cite{Ertl2024}.
If the registers are densely stored in a bit array, the theoretical \ac{MVP} for an efficient unbiased estimator meeting the Cramér-Rao bound is
\begin{equation}
	\label{equ:mvp_uncompressed}
	\textstyle
	\symMVP \approx
	\frac{(\symBitsForMax + \symNumExtraBits)\ln \symBase}{\symZetaFunc\!\left(2, 1 + \frac{\symBase^{-\symNumExtraBits}}{\symBase-1}\right)},
\end{equation}
where $\symZetaFunc$ is the Hurvitz zeta function as defined in \Cref{tab:notation}.
For \ac{ULL}, as a special case with $\symBase=2$ and $\symNumExtraBits=2$, we have shown that this \ac{MVP} can also be practically achieved using \ac{ML} estimation. We have also derived a first-order bias correction \cite{Cox1968} that can be applied to the \ac{ML} estimate $\symCardinalityEstimatorML$ according to
\begin{equation}
	\label{equ:ml_bias_correction}
	\symCardinalityEstimator = \frac{\symCardinalityEstimatorML}{1 + \frac{\symBiasCorrectionFactorConstant}{\symNumReg}}
	\quad
	\text{with}\
	\textstyle
	\symBiasCorrectionFactorConstant := (\ln \symBase)
	\left(
	1 +2\frac{\symBase^{-\symNumExtraBits}}{\symBase-1}
	\right)
	\frac{
		\symZetaFunc\!\left(3, 1 + \frac{\symBase^{-\symNumExtraBits}}{\symBase-1}\right)}{\left(\symZetaFunc\!\left(2, 1 + \frac{\symBase^{-\symNumExtraBits}}{\symBase-1}\right)\right)^2}.
\end{equation}
$\symBiasCorrectionFactorConstant$ is a constant dependent on the parameters. The bias correction factor $(1 + \frac{\symBiasCorrectionFactorConstant}{\symNumReg})^{-1}$ approaches 1 as $\symNumReg\rightarrow\infty$.

If the state size is measured in terms of the Shannon entropy, corresponding to optimal compression of the state, a different expression was obtained for the \ac{MVP}
\begin{equation}
	\label{equ:mvp_compressed}
	\textstyle
	\symMVP
	\approx
	\frac{
		\left(1+\frac{\symBase^{-\symNumExtraBits}}{\symBase-1}\right)^{\!-1}+
		\int_{0}^1 \symZ^{\frac{\symBase^{-\symNumExtraBits}}{\symBase-1}}
		\frac{
			(1-\symZ)
			\ln(1-\symZ)
		}{\symZ\ln \symZ}
		d\symZ
	}{
		\symZetaFunc\!\left(2, 1 + \frac{\symBase^{-\symNumExtraBits}}{\symBase-1}\right)\ln 2 }
	.
\end{equation}
This number is also known as \ac{FISH} number, for which a theoretical lower bound of 1.98 was postulated \cite{Pettie2021}.

In a non-distributed setting, where merging of sketches is not needed, martingale estimation can be used \cite{Ting2014,Cohen2015}. This estimation method is known to be efficient \cite{Pettie2021a} and leads to a smaller \ac{MVP} than \eqref{equ:mvp_uncompressed}. The corresponding asymptotic \ac{MVP}, if again the registers are densely stored in a bit array, is given by
\begin{equation}
	\label{equ:mvp_martingale}
	\textstyle
	\symMVP
	\approx
	\frac{(\symBitsForMax + \symNumExtraBits)\ln\symBase}{2}\left(1+\frac{\symBase^{-\symNumExtraBits}}{\symBase-1}\right).
\end{equation}
In contrast, if the registers are optimally compressed, the \ac{MVP} would be
\begin{equation}
	\label{equ:martingale_compressed}
	\textstyle
	\symMVP \approx
	\frac{
		1+
		\left(1+\frac{\symBase^{-\symNumExtraBits}}{\symBase-1}\right)
		\int_{0}^1 \symZ^{\frac{\symBase^{-\symNumExtraBits}}{\symBase-1}}
		\frac{
			(1-\symZ)
			\ln(1-\symZ)
		}{\symZ\ln \symZ}
		d\symZ
	}{
		2 \ln 2 }.
\end{equation}
This expression has a lower bound of 1.63 which corresponds to the theoretical limit \cite{Pettie2021a}.

\subsection{Approximated Update Value Distribution}

\begin{figure}
	\centering
	\includegraphics[width=\linewidth]{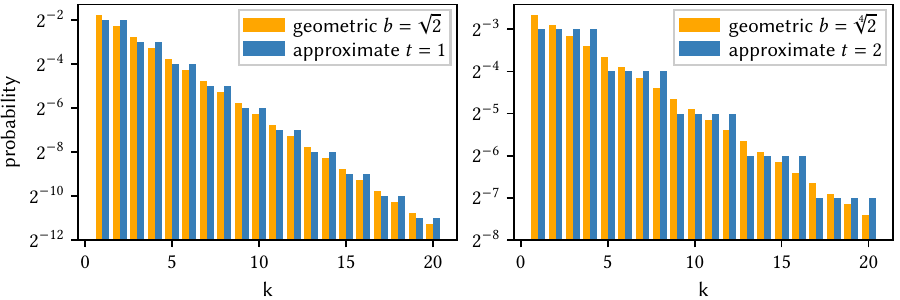}
	\caption{\boldmath Comparing the \acsp*{PMF} \eqref{equ:geometric} and \eqref{equ:step_dist} for $\symBase = 2^{2^{-\symExtraBits}}$.}
	\label{fig:dist_approx}
\end{figure}

Standard hash functions usually give a uniformly distributed 64-bit hash value. Getting a hash value that is geometrically distributed according to \eqref{equ:geometric} is not straightforward except for $\symBase=2$, where the \acf{NLZ} of a uniformly distributed hash value can be taken. For $\symBase\neq 2$, the hash value could be used to seed a pseudo-random number generator that generates values according to the desired distribution. However, this involves floating-point operations, which might be slow on devices without floating-point units. Alternatively, hash ranges that map to the same update value can be precomputed, which allows finding the update value via linear or binary search.

Both approaches will be slower and not branch-free, unlike for $\symBase=2$. They may also introduce inaccuracies for huge distinct counts, as the probabilities of update values will always be multiples of $2^{-64}$ when using 64-bit hashes.

Therefore, we propose to use a different update value distribution with \ac{PMF}
\begin{equation}
	\label{equ:step_dist}
	\textstyle
	\symDensityUpdate(\symUpdateVal) =
	\frac{1}{2^{\symExtraBits+1+\lfloor(\symUpdateVal-1)/2^\symExtraBits\rfloor}}
	\qquad\text{with $\symUpdateVal\geq 1, \symExtraBits\geq 0$}.
\end{equation}
This distribution approximates geometric distributions \eqref{equ:geometric} with base value $\symBase = 2^{2^{-\symExtraBits}}$ quite well as shown in \Cref{fig:dist_approx}. The reason is that chunks of $2^\symExtraBits$ subsequent update values have always the same total probability. In other words, $\sum_{\symUpdateVal=\symSomeConstant2^\symExtraBits +1}^{\symSomeConstant2^\symExtraBits+2^\symExtraBits}\symDensityUpdate(\symUpdateVal)=\frac{1}{2^{\symSomeConstant+1}}$ holds for \eqref{equ:geometric} and \eqref{equ:step_dist} and for all $\symSomeConstant\geq 0$.

A big advantage of \eqref{equ:step_dist} over \eqref{equ:geometric} is that update values can be easily and accurately generated from a 64-bit hash value using a few CPU instructions by taking $\symExtraBits$ bits of the hash value and determining the \ac{NLZ} of the remaining bits according to
\begin{equation}
	\label{equ:update_val}
	\text{update value}= \text{\acs*{NLZ}}\times 2^\symExtraBits + (\text{value of the $\symExtraBits$ bits})+1.
\end{equation}
Furthermore, as we will see in \Cref{sec:ml_estimation}, this distribution simplifies \acl{ML} estimation a lot compared to a geometric distribution with $\symBase\neq2$ thanks to its power-of-two probabilities.

\subsection{\acl*{ELL}}
\label{sec:exaloglog}

\myAlg[t]{
	\caption{Inserts an element with 64-bit hash value $\langle\symHashBit_{63} \symHashBit_{62} \ldots \symHashBit_{0}\rangle_2$ into an \acl*{ELL} with ($6+\symExtraBits+\symNumExtraBits$)-bit registers $\symRegister_0, \symRegister_1, \ldots, \symRegister_{\symNumReg-1}$ ($\symNumReg = 2^\symPrecision$) and initial values $\symRegister_\symRegAddr = 0$.}
	\label{alg:insertion}
	$\symRegAddr\gets \langle\symHashBit_{\symPrecision+\symExtraBits-1} \symHashBit_{\symPrecision+\symExtraBits-2}\ldots\symHashBit_{\symExtraBits}\rangle_2$\Comment*[r]{extract register index}
	$\symMasked\gets \langle\symHashBit_{63}\symHashBit_{62} \ldots\symHashBit_{\symPrecision+\symExtraBits}\!\underbracket[0.5pt][1pt]{1\ldots 1}_{\scriptscriptstyle\symPrecision+\symExtraBits}\rangle_2$\;
	$\symUpdateVal \gets \symNLZ(\symMasked)2^\symExtraBits +\langle\symHashBit_{\symExtraBits-1} \symHashBit_{\symExtraBits-2}\ldots\symHashBit_{0}\rangle_2+ 1$\Comment*[r]{\minibox[t]{compute update value \eqref{equ:update_val},\\$\symUpdateVal \in [1, (65-\symPrecision-\symExtraBits)2^\symExtraBits]$}}
	$\symMaxUpdateVal\gets\lfloor\symRegister_\symRegAddr/2^\symNumExtraBits\rfloor$\Comment*[r]{get max. update value from $\symRegister_\symRegAddr$, right-shift by $\symNumExtraBits$ bits}
	$\symDelta\gets \symUpdateVal - \symMaxUpdateVal$\;
	\uIf{$\symDelta>0$}{
		$\symRegister_\symRegAddr \gets\symUpdateVal\cdot2^\symNumExtraBits+ \lfloor(2^{\symNumExtraBits} + (\symRegister_\symRegAddr \bmod 2^\symNumExtraBits)) / 2^{\symDelta }\rfloor$\;
	}\ElseIf{$\symDelta<0 \wedge\symNumExtraBits+\symDelta\geq 0$}{
		$\symRegister_\symRegAddr \gets \symRegister_\symRegAddr \symBitwiseOr 2^{\symNumExtraBits+\symDelta}$\Comment*[r]{$\symBitwiseOr$ denotes bitwise OR}
	}
}

Our new data structure \acf{ELL} is based on the generalized data structure as we have recently proposed in \cite{Ertl2024} and also briefly discussed in \Cref{sec:data_structure}, however, using the approximate distribution \eqref{equ:step_dist} instead of the geometric distribution \eqref{equ:geometric} for the update values. Hence, the parameter $\symBase$ is replaced by $\symExtraBits$ were the two distributions are similar for $\symBase = 2^{2^{-\symExtraBits}}$. For $\symExtraBits=0 \Leftrightarrow \symBase=2$ the two distributions are even identical.

As the generalized data structure, \ac{ELL} consists of $\symNumReg=2^\symPrecision$ registers. Every time an element is inserted, a 64-bit hash value is computed of which $\symPrecision$ bits are used to select one of the registers. The remaining $64-\symPrecision$ bits are used to generate an update value according to \eqref{equ:step_dist}. Each register consists of $\symBitsForMax+\symNumExtraBits$ bits. The first $\symBitsForMax$ bits store the maximum update value $\symMaxUpdateVal$ seen so far for that register. The remaining $\symNumExtraBits$ bits memorize the occurrences of update values in the range $[\symMaxUpdateVal-\symNumExtraBits, \symMaxUpdateVal-1]$.

To live up its name, \ac{ELL} should, like \ac{HLL} with 6-bit registers \cite{Heule2013}, support distinct counts up to the exa-scale, which is sufficient for any conceivable practical use case.
As the maximum supported distinct count is roughly given by $\symBase^{2^\symBitsForMax}$ \cite{Ertl2024}, $\symBitsForMax= 6 + \symExtraBits$ must be chosen in order to get $\symBase^{2^\symBitsForMax} = (2^{2^{-\symExtraBits}})^{2^{6+\symExtraBits}}=2^{64}\approx \num{1.8e19}$.

The update procedure for inserting an element is summarized by \Cref{alg:insertion} and exemplified in \Cref{fig:running-example}. It splits the hash value into three parts. The first $64 - \symPrecision - \symExtraBits$ bits are used to determine the \acf{NLZ} which is therefore in the range $[0,64 - \symPrecision - \symExtraBits]$. The next $\symPrecision$ bits address the register, and the remaining last $\symExtraBits$ bits, in combination with the \ac{NLZ} of the first part, determine the update value according to \eqref{equ:update_val}. It is possible to use the bits in a different order. In particular, the three parts consisting of $64 - \symPrecision - \symExtraBits$, $\symPrecision$, and $\symExtraBits$ bits could be permuted. However, only if the $64 - \symPrecision - \symExtraBits$ bits for the \ac{NLZ} and the $\symPrecision$ bits for the address are adjacent and in this order, \ac{ELL} will be reducible, as described later in \Cref{sec:reducibility}.

Since \ac{ELL} uses 64-bit hashes, update values computed according to \eqref{equ:update_val} are limited by $(64 - \symPrecision-\symExtraBits)2^\symExtraBits + (2^\symExtraBits-1) + 1 = (65 - \symPrecision - \symExtraBits)2^\symExtraBits$.
For $\symPrecision \geq 2$, all possible update values thus fit into $6+\symExtraBits$ bits as $(65 - \symPrecision - \symExtraBits)2^\symExtraBits \leq 63\cdot 2^\symExtraBits \leq 2^{6+\symExtraBits}-1$ holds in any case. Adapting distribution \eqref{equ:step_dist} to incorporate the limitation of update values to the range $[1, (65 - \symPrecision - \symExtraBits)2^\symExtraBits]$ gives
\begin{equation}
	\label{equ:update_density}
	\textstyle
	\symDensityUpdate(\symUpdateVal)
	=
	\frac{1}{2^{\symExpFunc(\symUpdateVal)}}
	\qquad\text{with $\symUpdateVal\in[1, (65 - \symPrecision - \symExtraBits)2^\symExtraBits]$, $\symExtraBits\geq 0$,}
\end{equation}
where we introduced the function
\begin{equation}
	\label{equ:exponent_func}
	\symExpFunc(\symUpdateVal)
	:=
	\min(\symExtraBits + 1+ \lfloor (\symUpdateVal-1)/2^\symExtraBits \rfloor,64-\symPrecision).
\end{equation}

\begin{figure}[t]
	\centering
	\includegraphics[width=\columnwidth]{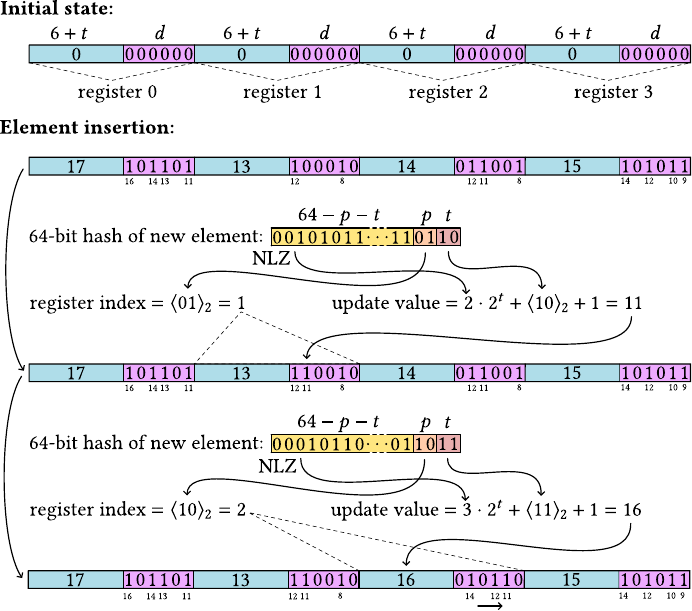}
	\caption{\boldmath Two element insertions into an \acl*{ELL} sketch with parameters $\symPrecision=2$, $\symExtraBits=2$, $\symNumExtraBits=6$ which has $2^\symPrecision=4$ registers with a size of $6+\symExtraBits+\symNumExtraBits=14$ bits.}
	\label{fig:running-example}
\end{figure}

\subsection{Choice of Parameters}
\label{sec:choice_param}

\begin{figure*}
	\centering
	\begin{minipage}{0.48\linewidth}
		\centering
		\includegraphics[width=\linewidth]{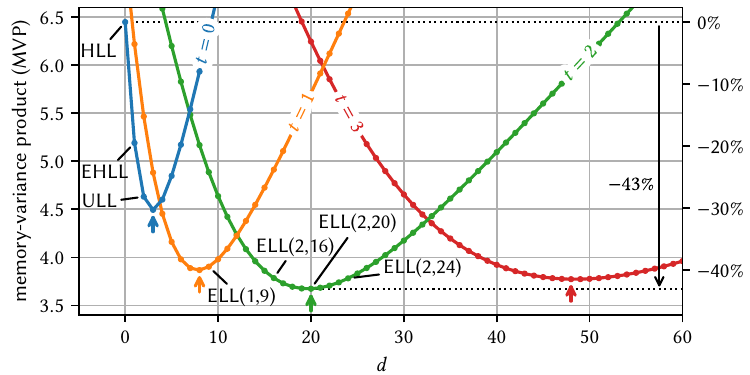}
		\caption{\boldmath The \acs*{MVP} according to \eqref{equ:mvp_uncompressed} with $\symBase = 2^{2^{-\symExtraBits}}$ and $\symBitsForMax= 6 + \symExtraBits$ when storing the registers in a bit array and using an efficient unbiased estimator. Arrows indicate minima.}
		\label{fig:mvp_ml}
	\end{minipage}
	\hfill
	\begin{minipage}{0.48\linewidth}
		\centering
		\includegraphics[width=\linewidth]{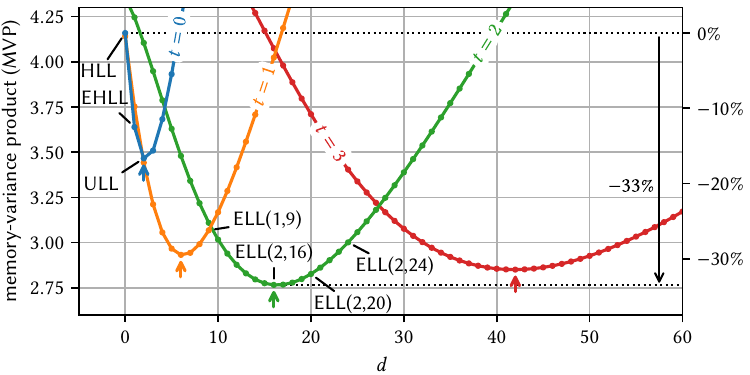}
		\caption{\boldmath The \acs*{MVP} according to \eqref{equ:mvp_martingale} with $\symBase = 2^{2^{-\symExtraBits}}$ and $\symBitsForMax= 6 + \symExtraBits$ when storing the registers in a bit array and using the martingale estimator. Arrows indicate minima.}
		\label{fig:mvp_martingale}
	\end{minipage}
\end{figure*}

\begin{figure*}
	\centering
	\begin{minipage}{0.48\linewidth}
		\centering
		\includegraphics[width=\linewidth]{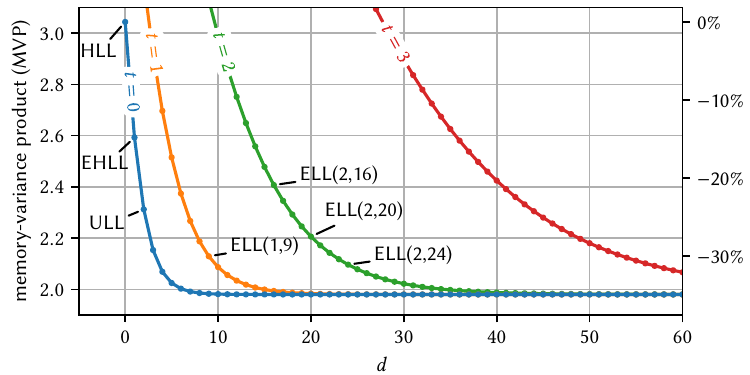}
		\caption{\boldmath The \acs*{MVP} according to \eqref{equ:mvp_compressed} with $\symBase = 2^{2^{-\symExtraBits}}$ when assuming optimal compression of the state and using an efficient unbiased estimator.}
		\label{fig:mvp_ml_compressed}
	\end{minipage}
	\hfill
	\begin{minipage}{0.48\linewidth}
		\centering
		\includegraphics[width=\linewidth]{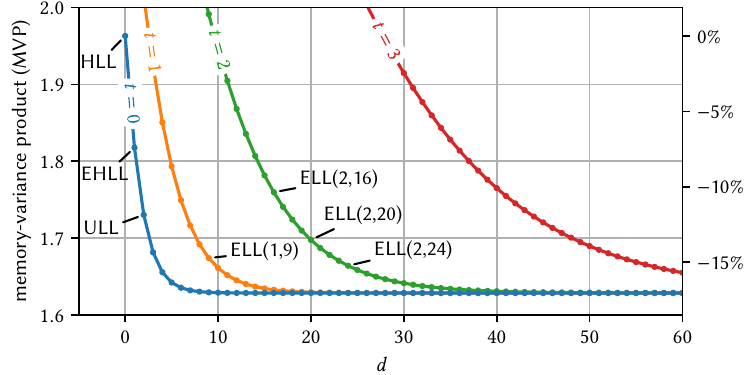}
		\caption{\boldmath The \acs*{MVP} according to \eqref{equ:martingale_compressed} with $\symBase = 2^{2^{-\symExtraBits}}$ when assuming optimal compression of the state and using the martingale estimator.}
		\label{fig:mvp_martingale_compressed}
	\end{minipage}
\end{figure*}

Under the assumption that the formulas for the \ac{MVP} presented in \Cref{sec:prev_results} are still a good approximation after the exchange of the update value distribution, we can use them to search for parameters that lead to a good space efficiency. As our experimental results will show later, the predicted \acp{MVP} are indeed very accurate despite the deviation from a geometric distribution which eventually justifies the made assumption.

We evaluated the \acp{MVP} given in \Cref{sec:prev_results} with $\symBitsForMax= 6 + \symExtraBits$ and $\symBase = 2^{2^{-\symExtraBits}}$ for $\symExtraBits\in\lbrace 0,1,2,3\rbrace$ and $\symNumExtraBits\geq 0$.
\Cref{fig:mvp_ml} shows the \ac{MVP} according to \eqref{equ:mvp_uncompressed} for an unbiased and efficient estimator.
\Cref{fig:mvp_martingale} shows the \ac{MVP} given by \eqref{equ:mvp_martingale} when using the martingale estimator that can only be used for non-distributed setups.
\Cref{fig:mvp_ml_compressed,fig:mvp_martingale_compressed} show the corresponding \acp{MVP} when the registers are assumed to be optimally compressed as given by \eqref{equ:mvp_compressed} and \eqref{equ:martingale_compressed}.

These figures allow identifying useful configurations. The optimal setting in \Cref{fig:mvp_ml} is $\symExtraBits=2$ and $\symNumExtraBits=20$, resulting in a total register size of $\symBitsForMax+\symNumExtraBits = 6 + \symExtraBits + \symNumExtraBits = 28$ bits and a theoretical \ac{MVP} of $3.67$ which is \qty{43}{\percent} less than that of \ac{HLL} with 6-bit registers. Since two registers can be packed into exactly 7 bytes, register access is not too complicated.

An interesting configuration is also $\symExtraBits=2$ and $\symNumExtraBits=24$ despite the larger theoretical \ac{MVP} with a value of $3.78$. The registers with a size of 32 bits allow very fast register access when stored in a 32-bit integer array. The 32-bit register alignment makes this configuration even convenient for concurrent updates using compare-and-swap instructions available on modern \acp{CPU}.
Furthermore, as \Cref{fig:mvp_ml_compressed,fig:mvp_martingale_compressed} indicate, this configuration is probably more efficient than $\symNumExtraBits=20$ or $\symNumExtraBits=16$ when using compression.

When setting $\symExtraBits=1$, the choice $\symNumExtraBits=9$ with a \ac{MVP} of $3.90$ is also worth mentioning. Although less space-efficient than the mentioned configurations with $\symExtraBits=2$, this setting also results in byte-aligned registers, as their size is exactly 16 bits. Variants with $\symExtraBits\geq 3$ lead to larger \acp{MVP} than for $\symNumExtraBits=2$. However, as they also lead to large $\symNumExtraBits$ and therefore to quite large register sizes, we do not consider them to be useful in practice.

In the case of martingale estimation, which can be used for non-distributed setups, the optimum is achieved for $\symExtraBits=2$ and $\symNumExtraBits=16$ (cf. \Cref{fig:mvp_martingale}), which leads to a \ac{MVP} of $2.77$, which is \qty{33}{\percent} below that of \ac{HLL} with 6-bit registers. As the register size is 24 bits and therefore fits exactly into 3 bytes, register access is also relatively simple.

As the theoretical \acp{MVP} only depend on $\symExtraBits$ and $\symNumExtraBits$, we will use the notation \acs*{ELL}($\symExtraBits$, $\symNumExtraBits$) to describe a specific class of \ac{ELL} sketches with equal \ac{MVP}. The third parameter, the precision parameter $\symPrecision$, controls the trade-off between space and estimation error.

\subsection{Relationship to Other Data Structures}
\label{sec:related_data_structures}
In \Cref{sec:exaloglog} we already mentioned that \ac{ELL} with $\symExtraBits=0$ is the same as the generalized data structure with $\symBase=2$ we have recently proposed \cite{Ertl2024}.
As this data structure is in turn a generalization of \ac{HLL} with $\symNumExtraBits=0$, \ac{EHLL} with $\symNumExtraBits=1$, and \ac{ULL} with $\symNumExtraBits=2$ \cite{Ertl2024}, these therefore correspond to \acs*{ELL}(0, 0), \acs*{ELL}(0, 1), and \acs*{ELL}(0, 2), respectively.
There is also a direct relationship to \ac{PCSA} \cite{Flajolet1985} and \ac{CPC} \cite{ApacheDataSketches, Lang2017}, as \acs*{ELL}(0, 64) stores exactly the same information, albeit encoded differently.
In addition, HyperMinHash \cite{Yu2022} corresponds to \acs*{ELL}($\symExtraBits$, 0), whose registers only store the maxima of update values. HyperMinHash uses an update value distribution equivalent to \eqref{equ:step_dist} but defines the ordering of register and update values differently.

\section{Statistical Inference}
\label{sec:statistical_inference}

To simplify the statistical model, we use the common Poisson approximation \cite{Flajolet2007, Ertl2017, Wang2023} that the number of inserted distinct elements is not fixed, but follows a Poisson distribution with mean $\symCardinality$. Consequently, since the updates are evenly distributed over all $\symNumReg$ registers, the number of updates with value $\symUpdateVal$ per register is again Poisson distributed with mean $\frac{\symCardinality}{\symNumReg}\symDensityUpdate(\symUpdateVal)$ when using $\symDensityUpdate$ defined in \eqref{equ:update_density}.
The probability that a register was updated with value $\symUpdateVal$ at least once, denoted by event $\symUpdateEvent_\symUpdateVal$, is therefore
\begin{equation}
	\label{equ:prob_update}
	\symProbability(\symUpdateEvent_\symUpdateVal)
	=
	1 - e^{-\frac{\symCardinality}{\symNumReg} \symDensityUpdate(\symUpdateVal)}.
\end{equation}

The probability that $\symMaxUpdateVal\in[1, (65 - \symPrecision - \symExtraBits)2^\symExtraBits]$ was the largest update value, which implies that there were no updates with values greater than $\symMaxUpdateVal$, is given by
\begin{multline}
	\label{equ:prob_max_update}
	\textstyle
	\symProbability( \symUpdateEvent_\symMaxUpdateVal\wedge \bigwedge_{\symUpdateVal=\symMaxUpdateVal+1}^{(65 - \symPrecision - \symExtraBits)2^\symExtraBits} \overline\symUpdateEvent_{\symUpdateVal})
	=
	\symProbability( \symUpdateEvent_\symMaxUpdateVal)
	\prod_{\symUpdateVal=\symMaxUpdateVal+1}^{(65 - \symPrecision - \symExtraBits)2^\symExtraBits}
	(1-\symProbability( \symUpdateEvent_\symUpdateVal))
	\\
	\textstyle
	=
	\left(1 - e^{
		-\frac{\symCardinality}{\symNumReg} \symDensityUpdate(\symMaxUpdateVal)}\right)
	\exp\!\left(
	-\frac{\symCardinality}{\symNumReg}
	\sum_{\symUpdateVal=\symMaxUpdateVal+1}^{(65 - \symPrecision - \symExtraBits)2^\symExtraBits}
	\symDensityUpdate(\symUpdateVal)
	\right)
	\\
	\textstyle
	=
	\left(1 - e^{
		-\frac{\symCardinality}{\symNumReg} \symDensityUpdate(\symMaxUpdateVal)}\right)
	e^{
			-\frac{\symCardinality}{\symNumReg}
			\symAlphaContribFunc(\symMaxUpdateVal) },
\end{multline}
where $\symAlphaContribFunc$ is defined as (see \Cref{lem:identity})
\begin{equation}
	\label{equ:alpha_contrib_func}
	\textstyle
	\symAlphaContribFunc(\symMaxUpdateVal)
	:=
	\sum_{\symUpdateVal = \symMaxUpdateVal+1}^{(65 - \symPrecision - \symExtraBits)2^\symExtraBits}\symDensityUpdate(\symUpdateVal)
	=
	\frac{2^{\symExtraBits}(1-\symExtraBits +\symExpFunc(\symMaxUpdateVal)) - \symMaxUpdateVal }{2^{\symExpFunc(\symMaxUpdateVal)}}.
\end{equation}

\subsection{Probability Mass Function for Registers}
\label{sec:pmf_reg}
As described in \Cref{sec:data_structure}, a register stores the maximum update value $\symMaxUpdateVal$ and also if there were updates with values in the range $[ \symMaxUpdateVal-\symNumExtraBits, \symMaxUpdateVal-1]$. Multiplying the probability \eqref{equ:prob_max_update} that $\symMaxUpdateVal$ was the maximum update value with $\symProbability(\symUpdateEvent_\symUpdateVal)$ or $(1-\symProbability(\symUpdateEvent_\symUpdateVal))$ (compare \eqref{equ:prob_update}) for all $\symUpdateVal\in[ \symMaxUpdateVal-\symNumExtraBits, \symMaxUpdateVal-1]$ dependent on whether an update with value $\symUpdateVal$ has occurred or not, gives the \acf{PMF} for a single register $\symRegister$:

\begin{description}
	\item[\boldmath Case $\symRegister=0$:]
		\begin{equation*}
			\textstyle
			\symDensityRegister(\symRegister\vert\symCardinality)=e^{-\frac{\symCardinality}{\symNumReg}\symAlphaContribFunc(0)} = e^{-\frac{\symCardinality}{\symNumReg}}.
		\end{equation*}

	\item[\boldmath Case $\symRegister = \symMaxUpdateVal 2^\symNumExtraBits + \langle\symIndexBit_1\ldots\symIndexBit_{\symMaxUpdateVal-1}\rangle_2 2^{\symNumExtraBits + 1- \symMaxUpdateVal}$ with $1 \leq \symMaxUpdateVal\leq \symNumExtraBits$:]
		\begin{multline*}
			\textstyle
			\symDensityRegister(\symRegister\vert\symCardinality) =
			(1 - e^{
					-\frac{\symCardinality}{\symNumReg} \symDensityUpdate(\symMaxUpdateVal)})
			e^{
					-\frac{\symCardinality}{\symNumReg}
					\symAlphaContribFunc(\symMaxUpdateVal) }
			\cdot
			\\
			\textstyle
			\cdot
			\prod_{\symIndexJ = 1}^{\symMaxUpdateVal-1}
			(1 - e^{
					-\frac{\symCardinality}{\symNumReg} \symDensityUpdate(\symMaxUpdateVal-\symIndexJ)})
			^{ \symIndexBit_\symIndexJ}
			(e^{
					-\frac{\symCardinality}{\symNumReg} \symDensityUpdate(\symMaxUpdateVal-\symIndexJ)})^{1-\symIndexBit_\symIndexJ}.
		\end{multline*}

	\item[\boldmath Case $\symRegister=\symMaxUpdateVal 2^\symNumExtraBits + \langle\symIndexBit_1\ldots\symIndexBit_{\symNumExtraBits}\rangle_2$ with $\symNumExtraBits+1 \leq \symMaxUpdateVal\leq (65 - \symPrecision - \symExtraBits)2^\symExtraBits$:]
		\begin{multline*}
			\textstyle
			\symDensityRegister(\symRegister\vert\symCardinality) =
			(1 - e^{
					-\frac{\symCardinality}{\symNumReg} \symDensityUpdate(\symMaxUpdateVal)})
			e^{
					-\frac{\symCardinality}{\symNumReg}
					\symAlphaContribFunc(\symMaxUpdateVal) }
			\cdot
			\\
			\textstyle
			\cdot
			\prod_{\symIndexJ = 1}^{\symNumExtraBits}
			(1 - e^{
					-\frac{\symCardinality}{\symNumReg} \symDensityUpdate(\symMaxUpdateVal-\symIndexJ)})
			^{ \symIndexBit_\symIndexJ}
			(e^{
					-\frac{\symCardinality}{\symNumReg} \symDensityUpdate(\symMaxUpdateVal-\symIndexJ)})^{1-\symIndexBit_\symIndexJ}.
		\end{multline*}
\end{description}

\subsection{Maximum-Likelihood Estimation}
\label{sec:ml_estimation}

Since the registers are independent due to the Poisson approximation, the log-likelihood function for register states $\symRegister_0, \symRegister_1,\ldots,\symRegister_{\symNumReg-1}$ can be written as
\begin{equation*}
	\textstyle
	\ln\symLikelihood =
	\ln\symLikelihood(\symCardinality\vert \symRegister_0,\ldots,\symRegister_{\symNumReg-1})
	=
	\sum_{\symRegAddr=0}^{\symNumReg-1} \ln \symDensityRegister(\symRegister_\symRegAddr\vert\symCardinality).
\end{equation*}
According to \eqref{equ:update_density}, $\symDensityUpdate$ is always a power of two from the set $\lbrace \frac{1}{2^{\symExtraBits+1}}, \frac{1}{2^{\symExtraBits+2}},\ldots,\frac{1}{2^{64-\symPrecision}}\rbrace$. Therefore, the log-likelihood function always has the shape
\begin{equation}
	\label{equ:log_likelihood_shape}
	\textstyle
	\ln\symLikelihood
	=-\frac{\symCardinality}{\symNumReg} \symLikelihoodFuncExponentOne+
	\sum_{\symMaxUpdateVal=\symExtraBits + 1}^{64 - \symPrecision}\symLikelihoodFuncExponentTwo_\symMaxUpdateVal \ln(1 - e^{-\frac{\symCardinality}{\symNumReg 2^\symMaxUpdateVal}} ),
\end{equation}
where the sum has at most $64-\symPrecision-\symExtraBits$ nonzero terms.
The coefficients $\symLikelihoodFuncExponentOne$ and $\symLikelihoodFuncExponentTwo_\symMaxUpdateVal$ depend just on the register states $\symRegister_0, \symRegister_1,\ldots,\symRegister_{\symNumReg-1}$ and can be computed according to \Cref{alg:coefficient_computations}. Since $\symIndexJ$ is the result of $\symExpFunc(\symUpdateVal)$ \eqref{equ:update_density}, it is bounded by $64-\symPrecision$ and all contributions to $\symLikelihoodFuncExponentOne$ are integer multiples of $\frac{1}{2^{64-\symPrecision}}$. Therefore, the summation can be performed with integer arithmetic only, if $\symLikelihoodFuncExponentOne' = \symLikelihoodFuncExponentOne\cdot 2^{64-\symPrecision}$ is computed instead.

\myAlg{
\caption{Computation of the coefficients of log-likelihood function \eqref{equ:log_likelihood_shape} for an \acl*{ELL} sketch with registers $\symRegister_0, \symRegister_1, \ldots, \symRegister_{\symNumReg-1}$.}
\label{alg:coefficient_computations}
$\symLikelihoodFuncExponentOne'\gets 0$\Comment*[r]{$\symLikelihoodFuncExponentOne'=\symLikelihoodFuncExponentOne\cdot2^{64-\symPrecision}$ is an integer}
$(\symLikelihoodFuncExponentTwo_{\symExtraBits+1}, \symLikelihoodFuncExponentTwo_{\symExtraBits+2}, \ldots, \symLikelihoodFuncExponentTwo_{64-\symPrecision})\gets (0, 0,\dots,0)$\;
\For(\Comment*[f]{iterate over all $\symNumReg$ registers}){$\symRegAddr\gets 0\ \KwTo\ \symNumReg-1$}{
$\symMaxUpdateVal\gets\lfloor\symRegister_\symRegAddr/2^\symNumExtraBits\rfloor$\Comment*[r]{get max. update value from $\symRegister_\symRegAddr$, right-shift by $\symNumExtraBits$ bits}
$\symLikelihoodFuncExponentOne'\gets\symLikelihoodFuncExponentOne' + 2^{64-\symPrecision}\symAlphaContribFunc(\symMaxUpdateVal)$\Comment*[r]{see \eqref{equ:alpha_contrib_func} for a definition of function $\symAlphaContribFunc$}
\If{$\symMaxUpdateVal \geq 1$}{
$\symIndexJ\gets \symExpFunc(\symMaxUpdateVal)$\Comment*[r]{see \eqref{equ:exponent_func} for a definition of function $\symExpFunc$}
$\symLikelihoodFuncExponentTwo_\symIndexJ\gets\symLikelihoodFuncExponentTwo_\symIndexJ + 1$\;
\If{$\symMaxUpdateVal \geq 2$}{
\For{$\symUpdateVal \gets \max(1, \symMaxUpdateVal-\symNumExtraBits)\ \KwTo \ \symMaxUpdateVal-1$} {
$\symIndexJ\gets \symExpFunc(\symUpdateVal)$\;
\uIf(\Comment*[f]{$\symBitwiseAnd$ denotes bitwise AND}){$(\symRegister_\symRegAddr \symBitwiseAnd 2^{\symNumExtraBits - \symMaxUpdateVal + \symUpdateVal}) = 0$ }{
$\symLikelihoodFuncExponentOne'\gets\symLikelihoodFuncExponentOne'+2^{64-\symPrecision-\symIndexJ}$\;
}
\Else{
	$\symLikelihoodFuncExponentTwo_\symIndexJ\gets\symLikelihoodFuncExponentTwo_\symIndexJ + 1$\;
}
}
}
}
}
$\symLikelihoodFuncExponentOne\gets\symLikelihoodFuncExponentOne'/2^{64-\symPrecision}$\;
}

In the following we consider the case where $\symLikelihoodFuncExponentOne$ and at least one of the coefficients $\symLikelihoodFuncExponentTwo_\symMaxUpdateVal$ are positive. $\symLikelihoodFuncExponentTwo_\symMaxUpdateVal$ all zero requires all registers to be in the initial state leading to a \ac{ML} estimate of zero. $\symLikelihoodFuncExponentOne=0$ can only occur, if all registers are saturated. In this case, which only occurs with a noteworthy probability for distinct counts that are entirely unrealistic, the \ac{ML} estimate would be infinite.
When introducing
\begin{align}
	\label{equ:defx}
	\symX
	 & := \textstyle\exp\!\left(\frac{\symCardinality}{\symNumReg 2^{\symMaxUpdateValMax} }\right) - 1,
	\\
	\symMaxUpdateValMin
	 & :=\textstyle\min_{\symExtraBits + 1\leq \symMaxUpdateVal\leq 64 - \symPrecision}\lbrace\symMaxUpdateVal\mid \symLikelihoodFuncExponentTwo_\symMaxUpdateVal > 0\rbrace,
	\nonumber
	\\
	\symMaxUpdateValMax
	 & :=\textstyle\max_{\symExtraBits + 1\leq \symMaxUpdateVal\leq 64 - \symPrecision} \lbrace\symMaxUpdateVal\mid \symLikelihoodFuncExponentTwo_\symMaxUpdateVal > 0\rbrace,
	\nonumber
	\\
	\label{equ:def_suma}
	\symSumA(\symX)
	 & :=
	\textstyle
	\symLikelihoodFuncExponentTwo_{\symMaxUpdateValMax}
	+
	\sum_{\symIndexJ=1}^{\symMaxUpdateValMax - \symMaxUpdateValMin}\symLikelihoodFuncExponentTwo_{\symMaxUpdateValMax - \symIndexJ}
	\prod_{\symIndexL=0}^{\symIndexJ-1}
	\frac{2}{(1+\symX)^{2^\symIndexL}+1},
\end{align}
the \ac{ML} equation can be equivalently written as
$\symFunc(\symX) = 0$ where $\symFunc$ is defined as
\begin{multline}
	\label{equ:ml_func}
	\textstyle
	\symFunc(\symX)
	:=
	-\symX(1+\symX)\frac{\partial}{\partial\symX}\ln\symLikelihood
	\\
	\textstyle
	=
	\symLikelihoodFuncExponentOne 2^{\symMaxUpdateValMax}\symX
	-
	\sum_{\symIndexJ=0}^{\symMaxUpdateValMax - \symMaxUpdateValMin} \frac{ \symLikelihoodFuncExponentTwo_{\symMaxUpdateValMax - \symIndexJ} 2^\symIndexJ \symX}{(1+\symX)^{2^\symIndexJ}-1}
	=
	\symLikelihoodFuncExponentOne 2^{\symMaxUpdateValMax}\symX
	-
	\symSumA(\symX).
\end{multline}
This function is strictly increasing and concave for $\symX\geq 0$ as shown in \Cref{lem:inc_and_concave}. As a result, it has a well-defined root, because $\symFunc(0)= -\sum_{\symIndexJ=\symMaxUpdateValMin}^{\symMaxUpdateValMax}\symLikelihoodFuncExponentTwo_\symIndexJ<0$ and $\symFunc(\infty) \rightarrow \infty > 0$.
If $\symXZero$ denotes the root of $\symFunc$, hence $\symFunc(\symXZero) = 0$, the \ac{ML} estimate $\symCardinalityEstimatorML$ is given according to \eqref{equ:defx} by
\begin{equation}
	\label{equ:final_estimate}
	\symCardinalityEstimatorML = \symNumReg 2^{\symMaxUpdateValMax} \ln(1+ \symXZero).
\end{equation}
As an optional last step, the \ac{ML} estimate can be corrected according to \eqref{equ:ml_bias_correction} with $\symBase=2^{2^{-\symExtraBits}}$ to reduce the bias.

The evaluation of $\symFunc$ is cheap, because the number of terms is limited by $\symMaxUpdateValMax-\symMaxUpdateValMin+1\leq 64-\symPrecision-\symExtraBits$. Moreover, as all occurring exponents are powers of two, they can be computed recursively by squaring using simple multiplications. As the denominator in \eqref{equ:def_suma} is always greater than or equal to 2, the evaluation is also numerically safe. To reduce the numerical error, we can replace the denominator by $2 +\symY_\symIndexL$ with
\begin{equation}
	\label{equ:y_def}
	\symY_\symIndexL := (1+\symX)^{2^\symIndexL}-1 \geq 0
\end{equation}
and use the recursion
\begin{equation}
	\label{equ:y_recursion}
	\symY_{\symIndexL+1}=\symY_{\symIndexL}\cdot(2 + \symY_{\symIndexL}).
\end{equation}
The product appearing in \eqref{equ:def_suma}
\begin{equation*}
	\textstyle
	\symFactorOne_\symIndexJ:=
	\prod_{\symIndexL=0}^{\symIndexJ-1}
	\frac{2}{(1+\symX)^{2^\symIndexL}+1}
\end{equation*}
can then be also computed recursively following
\begin{equation}
	\label{equ:factor_one_def}
	\textstyle
	\symFactorOne_{\symIndexJ+1}
	=
	\symFactorOne_{\symIndexJ} \frac{2}{2+\symY_{\symIndexJ}}.
\end{equation}

The simplicity of the \ac{ML} equation is a result of distribution \eqref{equ:step_dist}.
For comparison, geometrically distributed update values following \eqref{equ:geometric} would have led to significantly more terms. In addition, the computations of the resulting power expressions with real exponents would have been much more expensive. \Cref{app:numerical-root-finding} describes a robust and fast-converging algorithm for finding the root of the \ac{ML} equation.

\subsection{Martingale Estimation}
\label{sec:martingale_estimation}

A simple, efficient, and unbiased way to estimate the distinct count is to start from zero and increment the estimate, whenever the state of the sketch is modified, according to the inverse of the probability that such a modification occurs with the insertion of any unseen element. This online approach is known as martingale or \acf{HIP} estimation \cite{Ting2014, Cohen2015} and even leads to smaller estimation errors as already mentioned in \Cref{sec:prev_results,sec:choice_param}. However, martingale estimation is limited to cases where the data is not distributed and merging of sketches is not needed.

In addition to the estimate, the martingale estimator also keeps track of the current state change probability $\symStateChangeProbability$. Initially, $\symStateChangeProbability=1$ as the first update will certainly change the state. Whenever a register is modified, the probability of state changes for further elements decreases. The probability, that a new unseen element changes the \ac{ELL} state, is given by
\begin{equation}
	\label{equ:state_change_probability}
	\textstyle
	\symStateChangeProbability(\symRegister_0, \ldots, \symRegister_{\symNumReg-1}) = \sum_{\symRegAddr=0}^{\symNumReg-1} \symRegMartingale(\symRegister_\symRegAddr).
\end{equation}
$\symRegMartingale(\symRegister_\symRegAddr)$ is the probability that register $\symRegister_\symRegAddr$ is changed with the next new element. The function $\symRegMartingale$ is strictly monotonically decreasing and defined for a register value $\symRegister = \symMaxUpdateVal 2^\symNumExtraBits + \langle\symIndexBit_1\ldots\symIndexBit_{\symNumExtraBits}\rangle_2$ as
\begin{equation*}
	\textstyle
	\symRegMartingale(\symRegister)
	=
	\frac{1}{\symNumReg}(\symAlphaContribFunc(\symMaxUpdateVal)
	+
	\sum_{\symUpdateVal = \max(1, \symMaxUpdateVal-\symNumExtraBits)}^{\symMaxUpdateVal-1}
	(1 - \symIndexBit_{\symMaxUpdateVal-\symUpdateVal})\cdot\symDensityUpdate(\symUpdateVal))
\end{equation*}
when using $\symAlphaContribFunc$ from \eqref{equ:alpha_contrib_func} and $\symDensityUpdate$ from \eqref{equ:update_density}.

\myAlg{
	\caption{Incrementally updates the martingale estimate $\symCardinalityEstimatorMartingale$ and the state change probability $\symStateChangeProbability$ whenever a register is altered from $\symRegister$ to $\symRegister'$ ($\symRegister<\symRegister'$). Initially, $\symCardinalityEstimatorMartingale = 0$ and $\symStateChangeProbability=1$.}
	\label{alg:martingale}
	$\symCardinalityEstimatorMartingale\gets \symCardinalityEstimatorMartingale + \frac{1}{\symStateChangeProbability}$\Comment*[r]{update estimate}
	$\symStateChangeProbability\gets \symStateChangeProbability - \left(\symRegMartingale(\symRegister) - \symRegMartingale(\symRegister')\right)$ \Comment*[r]{\minibox[t]{update state change probability,\\$\symRegMartingale(\symRegister)>\symRegMartingale(\symRegister')$, compare \eqref{equ:state_change_probability}}}
}

The martingale estimator is incremented with every state change by $\frac{1}{\symStateChangeProbability}$ prior the update as demonstrated by \Cref{alg:martingale}. $\symStateChangeProbability$ itself can also be incrementally adjusted, such that the whole update takes constant time. The martingale estimator is unbiased and optimal, if mergeability is not needed \cite{Pettie2021a}.

\section{Practical Implementation}

\myAlg{
\caption{Merges two corresponding registers $\symRegister$ and $\symRegister'$ of \acl*{ELL} sketches with identical parameters $\symExtraBits$, $\symNumExtraBits$, and $\symPrecision$.}
\label{alg:merge_register}
\Fn{\FuncMergeReg{$\symRegister$, $\symRegister'$, $\symNumExtraBits$}}{
$\symMaxUpdateVal\gets\lfloor\symRegister/2^{\symNumExtraBits} \rfloor$\Comment*[r]{get max. update value from $\symRegister$, right-shift by $\symNumExtraBits$ bits}
$\symMaxUpdateVal'\gets\lfloor\symRegister'/2^{\symNumExtraBits} \rfloor$\Comment*[r]{get max. update value from $\symRegister'$, right-shift by $\symNumExtraBits$ bits}
\uIf{$\symMaxUpdateVal>\symMaxUpdateVal' \wedge \symMaxUpdateVal' > 0$}{
	\KwRet $\symRegister \symBitwiseOr \lfloor(2^{\symNumExtraBits} + (\symRegister' \bmod 2^{\symNumExtraBits} )) / 2^{\symMaxUpdateVal - \symMaxUpdateVal'}\rfloor$\Comment*[r]{$\symBitwiseOr$ denotes bitwise OR}
}\uElseIf{$\symMaxUpdateVal'>\symMaxUpdateVal\wedge \symMaxUpdateVal > 0$}{
	\KwRet $\symRegister'\symBitwiseOr\lfloor(2^{\symNumExtraBits} + (\symRegister \bmod 2^{\symNumExtraBits} ) ) / 2^{\symMaxUpdateVal' - \symMaxUpdateVal}\rfloor$\;
}\Else{
	\KwRet $\symRegister\symBitwiseOr\symRegister'$\;
}
}
}

Like other probabilistic data structures such as \ac{HLL}, \ac{ELL} also relies on high-quality 64-bit hash values for the elements. Known good hash functions are WyHash \cite{Yi}, Komihash \cite{Vaneev}, or PolymurHash \cite{Peters}.
Insert operations according to \Cref{alg:insertion} obviously take constant time and are very fast, because all statements can be translated into inexpensive \ac{CPU} instructions. Expressions of kind $\lfloor \symX / 2^\symY\rfloor$ can be realized by a right-shift operation by $\symY$ bits, and $\symX \bmod 2^\symY$ is the same as masking the lower $\symY$ bits. Furthermore, specific instructions exist on modern \acp{CPU} to obtain the \acf{NLZ}.

\subsection{Mergeability}
If two \ac{ELL} data structures are equally configured, thus they have equal $\symPrecision$, $\symExtraBits$, and $\symNumExtraBits$ values, they can be easily merged by pairwise merging of individual registers. A register stores the maximum update value $\symMaxUpdateVal$ in its upper $6+\symExtraBits$ and the occurrences of update values in the range $[\symMaxUpdateVal- \symNumExtraBits,\symMaxUpdateVal-1]$ in its lower $\symNumExtraBits$ bits (compare \Cref{fig:running-example}). Since the merged state is the result of the union of all updates, the merged register must finally store the common maximum update value, and the $\symNumExtraBits$ bits must indicate the combined occurrences of the next $\symNumExtraBits$ smaller update values relative to this common maximum. \Cref{alg:merge_register} demonstrates how efficient register merging can be realized using bitwise operations. The result of the merge operation is equivalent to inserting the union of all individual original elements, previously inserted into one of the two data structures, directly into an empty data structure using \Cref{alg:insertion}.

\ac{ELL} sketches are also mergeable if not all the parameters are equal as long as the sketches share the same $\symExtraBits$-parameter. If the sketch parameters are $(\symExtraBits, \symNumExtraBits, \symPrecision)$ and $(\symExtraBits, \symNumExtraBits', \symPrecision')$, respectively, they both can be reduced to an \ac{ELL} sketch with parameters $(\symExtraBits, \min(\symNumExtraBits,\symNumExtraBits'), \min(\symPrecision, \symPrecision'))$ first. This is useful for migration scenarios, if the precision $\symPrecision$ or parameter $\symNumExtraBits$ must be changed while mergeability with older records should be preserved.

\myAlg{
\caption{Reduces an \acf*{ELL} sketch with registers $\symRegister_0, \symRegister_1, \ldots, \symRegister_{\symNumReg-1}$ and parameters $\symExtraBits$, $\symNumExtraBits$, $\symPrecision$ to an \acs*{ELL} sketch with registers $\symRegister'_0, \symRegister'_1, \ldots, \symRegister'_{\symNumReg'-1}$ and parameters $\symExtraBits$, $\symNumExtraBits'$, $\symPrecision'$, where $\symNumReg = 2^\symPrecision$, $\symNumReg' = 2^{\symPrecision'}$, $\symNumExtraBits \geq \symNumExtraBits' \geq 0$, and $\symPrecision \geq \symPrecision' \geq 0$.
}
\label{alg:reduction}
$\symThreshold\gets (64 - \symExtraBits - \symPrecision)\cdot 2^\symExtraBits + 1$\;
\For{$\symRegAddr\gets 0\ \KwTo\ \symNumReg'-1$}{
$\symRegister'\gets 0$\;
\For{$\symIndexJ\gets 0\ \KwTo\ 2^{\symPrecision - \symPrecision'}-1$}{
$\symRegister\gets \lfloor\symRegister_{\symRegAddr + \symIndexJ \cdot \symNumReg'} / 2^{\symNumExtraBits - \symNumExtraBits'}\rfloor$\Comment*[r]{right-shift by $\symNumExtraBits-\symNumExtraBits'$ bits}
$\symMaxUpdateVal\gets\lfloor\symRegister/2^{\symNumExtraBits'} \rfloor$\Comment*[r]{right-shift by $\symNumExtraBits'$ bits}
\If(\Comment*[f]{satisfied if \ac*{NLZ} was $64- \symExtraBits - \symPrecision$ in \eqref{equ:update_val} for $\symMaxUpdateVal$}){$\symMaxUpdateVal \geq \symThreshold$}{
	\Comment*[r]{$\symRegister$ must be adapted, if $\symMaxUpdateVal$ was different at precision $\symPrecision'$}
	$\symShift\gets((\symPrecision - \symPrecision')-(64 - \symNLZ(\symIndexJ)))\cdot 2^\symExtraBits$\Comment*[r]{\minibox[t]{assuming $\symIndexJ$ has 64 bits\\$\Rightarrow \symNLZ(\symIndexJ)\in[0,64]$} }
	\If{$\symShift > 0$}{
		$\symNumBits\gets\symNumExtraBits' + \symThreshold - \symMaxUpdateVal$\;
		\lIf{$\symNumBits > 0$}{
			$\symRegister\gets \lfloor\symRegister/2^\symNumBits\rfloor\cdot 2^\symNumBits + \lfloor(\symRegister \bmod 2^\symNumBits)/2^\symShift\rfloor$
		}
		$\symRegister\gets\symRegister + \symShift\cdot2^{\symNumExtraBits'}$\;
	}
}
$\symRegister'\gets$ \FuncMergeReg{$\symRegister$, $\symRegister',\symNumExtraBits'$}\Comment*[r]{see \Cref{alg:merge_register}}
}
$\symRegister'_{\symRegAddr}\gets \symRegister'$\;
}
}

\subsection{Reducibility}
\label{sec:reducibility}
The reduction of the $\symNumExtraBits$-parameter is straightforward. Decrementing it from $\symNumExtraBits$ to $\symNumExtraBits'$ with $\symNumExtraBits\geq\symNumExtraBits'$ only requires right-shifting all registers by $\symNumExtraBits-\symNumExtraBits'$ bits.
The reduction of the precision from $\symPrecision$ to $\symPrecision'$ is more complex, as $2^{\symPrecision-\symPrecision'}$ registers need to be combined into one. However, due to the way in which the hash bits are consumed in \Cref{alg:insertion}, this is also possible in a lossless way, meaning that the result is the same as if all elements would be recorded directly using a sketch with the reduced parameters, as demonstrated by \Cref{alg:reduction}.

\subsection{Sparse Mode}
\label{sec:sparse_mode}

\acl{ELL} uses a fixed array of registers, which guarantees constant-time insertions. However, if space efficiency is more important, allocating that array from the beginning does not make sense if keeping the raw input data takes less space. Therefore, many data sketches start with a \emph{sparse} mode with a linearly scaling memory footprint and only switch to the \emph{dense} representation at the break-even point.

For \ac{ELL}, a sparse representation could be realized by just storing the 64-bit input hash values in a list. To save space, we can reduce those hash values to $(\symHashTokenParameter+6)$-bit values, which we call hash tokens, by keeping only information needed for insertions into \ac{ELL} sketches with $\symPrecision + \symExtraBits\leq \symHashTokenParameter$. A hash token stores the least significant $\symHashTokenParameter$ bits of the original hash value and, in addition, the \ac{NLZ} of the remaining $(64-\symHashTokenParameter)$ most significant bits of the hash value. If $\symHashTokenParameter\geq 1$ the \ac{NLZ} fits into 6 bits and a 64-bit hash value $\langle\symHashBit_{63} \symHashBit_{62} \ldots \symHashBit_{0}\rangle_2$ can be mapped to a $(\symHashTokenParameter+6)$-bit hash token $\symToken$ according to
\begin{equation*}
	\symToken = \langle \symHashBit_{\symHashTokenParameter-1}\dots \symHashBit_{0}000000\rangle_2 + \symNLZ(\langle\symHashBit_{63}\symHashBit_{62} \ldots\symHashBit_{\symHashTokenParameter}\!\underbracket[0.5pt][1pt]{11\ldots 1}_{\scriptscriptstyle\symHashTokenParameter}\rangle_2).
\end{equation*}

While in sparse mode, it is sufficient to keep only distinct hash tokens. When switching to dense mode, the hash tokens can be transformed back to representative 64-bit hash values following
\begin{equation*}
	\langle\symHashBit'_{63} \symHashBit'_{62} \ldots \symHashBit'_{0}\rangle_2
	=
	2^{64-\langle\symHashTokenBit_{5}\symHashTokenBit_{4}\symHashTokenBit_{3}\symHashTokenBit_{2}\symHashTokenBit_{1}\symHashTokenBit_{0}\rangle_2}
	- 2^{\symHashTokenParameter} + \langle\symHashTokenBit_{\symHashTokenParameter+5}\symHashTokenBit_{\symHashTokenParameter+4}\ldots \symHashTokenBit_{6}\rangle_2
\end{equation*}
where $\langle\symHashTokenBit_{\symHashTokenParameter+5}\symHashTokenBit_{\symHashTokenParameter+4}\ldots\symHashTokenBit_0\rangle_2$ is the binary representation of the token. The reconstructed hash values can be equivalently used for the insertion using \Cref{alg:insertion} as the original hash value.

It is also possible, to estimate the distinct count directly from a given set of distinct hash tokens $\symTokenSet$. Since the first $\symHashTokenParameter$ bits are uniformly distributed and the \ac{NLZ}, stored in the remaining 6 bits, are distributed according to a truncated geometric distribution with maximum value $64-\symHashTokenParameter$, the \acf{PMF} of hash tokens is given by
\begin{equation}
	\label{equ:pmf_hash_token}
	\symDensityToken(\symToken)
	=
	\begin{cases}
		\frac{1}{2^{\min(\symHashTokenParameter + 1 + (\symToken \bmod 64), 64)}}
		 &
		\symToken \bmod 64 \leq 64 - \symHashTokenParameter,
		\\
		0
		 &
		\text{else,}
	\end{cases}
\end{equation}
with $\symToken\in[0, 2^{\symHashTokenParameter+6})$ and $\symHashTokenParameter\geq 1$. As for any \ac{PMF}, summing up the probabilities for all possible values yields 1
\begin{equation}
	\label{equ:token_mass}
	\textstyle
	\sum_{\symToken=0}^{2^{\symHashTokenParameter+6}-1}\symDensityToken(\symToken) = 1.
\end{equation}

As in \Cref{sec:statistical_inference}, we use again the Poisson approximation, which allows to write the probability that some hash token $\symToken$ is in the set of collected hash tokens $\symTokenSet$ as
$\symProbability(\symToken\in\symTokenSet)=1-e^{-\symCardinality\symDensityToken(\symToken)}$.
Therefore, the log-likelihood function is
\begin{align}
	\ln\symLikelihood
	 & = \ln\symLikelihood(\symCardinality \vert \symTokenSet) \nonumber
	\\
	 &
	\textstyle
	=
	\sum_{\symToken\notin \symTokenSet} \ln(e^{-\symCardinality\symDensityToken(\symToken)})
	+
	\sum_{\symToken\in \symTokenSet} \ln(1
	-
	e^{-\symCardinality\symDensityToken(\symToken)})\nonumber
	\\
	 &
	\textstyle
	=
	-
	\symCardinality
	\sum_{\symToken\notin \symTokenSet} \symDensityToken(\symToken)
	+
	\sum_{\symToken\in \symTokenSet} \ln(
	1
	-
	e^{-\symCardinality\symDensityToken(\symToken)}
	)\nonumber
	\\
	 &
	\textstyle
	=
	-
	\symCardinality
	(1 - \sum_{\symToken\in \symTokenSet} \symDensityToken(\symToken))
	+
	\sum_{\symToken\in \symTokenSet} \ln(1
	-
	e^{-\symCardinality\symDensityToken(\symToken)})\nonumber
	\\
	\label{equ:token_ml}
	 &
	\textstyle
	=
	-\symCardinality \symLikelihoodFuncExponentOne + \sum_{\symMaxUpdateVal=\symHashTokenParameter+1}^{64}\symLikelihoodFuncExponentTwo_\symMaxUpdateVal \ln(1 - e^{-\frac{\symCardinality}{2^\symMaxUpdateVal}} )
\end{align}
where we used \eqref{equ:token_mass}. The coefficients $\symLikelihoodFuncExponentOne$ and $\symLikelihoodFuncExponentTwo_\symMaxUpdateVal$ can be obtained using \Cref{alg:token_coefficients}.
The log-likelihood function \eqref{equ:token_ml} has the same shape as that for the \ac{ELL} registers \eqref{equ:log_likelihood_shape} when setting $\symNumReg=1 \Leftrightarrow \symPrecision=0$ and $\symExtraBits=\symHashTokenParameter$. Therefore, the \ac{ML} estimate can be found again with the same root-finding algorithm described in \Cref{app:numerical-root-finding}.

\myAlg{
	\caption{Computation of the coefficients for the log-likelihood function \eqref{equ:token_ml} from a set $\symTokenSet$ of distinct $(\symHashTokenParameter+6)$-bit hash tokens.}
	\label{alg:token_coefficients}
	$\symLikelihoodFuncExponentOne'\gets 2^{64}$\Comment*[r]{start from 0 when using an unsigned 64-bit integer}
	$(\symLikelihoodFuncExponentTwo_{\symHashTokenParameter+1}, \symLikelihoodFuncExponentTwo_{\symExtraBits+2}, \ldots, \symLikelihoodFuncExponentTwo_{64})\gets (0, 0,\dots,0)$\;
	\For(\Comment*[f]{iterate over all collected distinct tokens}){$\symToken \in \symTokenSet$}{
		$\symIndexJ \gets \min(\symHashTokenParameter+1 + (\symToken \bmod 64), 64)$\;
		$\symLikelihoodFuncExponentTwo_\symIndexJ \gets \symLikelihoodFuncExponentTwo_\symIndexJ+ 1$\;
		$\symLikelihoodFuncExponentOne'\gets \symLikelihoodFuncExponentOne' - 2^{64-\symIndexJ}$\;
	}
	$\symLikelihoodFuncExponentOne\gets \symLikelihoodFuncExponentOne' / 2^{64}$;
}

\section{Experiments}
\label{sec:experiments}

\begin{figure*}[t]
	\centering
	\includegraphics[width=\linewidth]{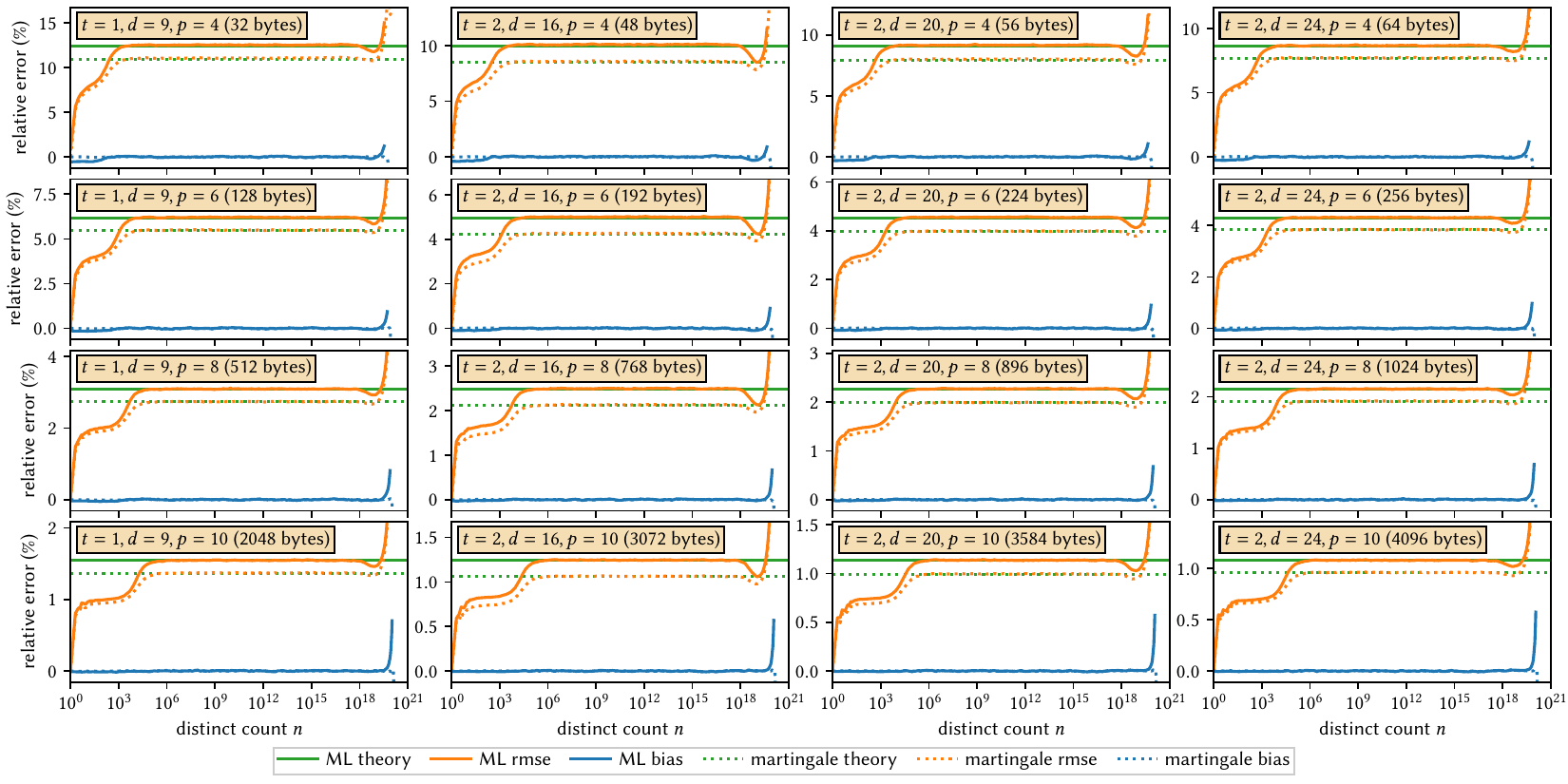}
	\caption{\boldmath The relative bias and the \acs*{RMSE} for the \acs*{ML} and the martingale estimator for different \acs*{ELL} configurations obtained from \num{100000} simulation runs. The theoretically predicted errors perfectly match the experimental results. Individual insertions were simulated up to a distinct count of \num{e6} before switching to the fast simulation strategy.}
	\label{fig:estimation-error}
\end{figure*}

\begin{figure}[t]
	\centering
	\includegraphics[width=\columnwidth]{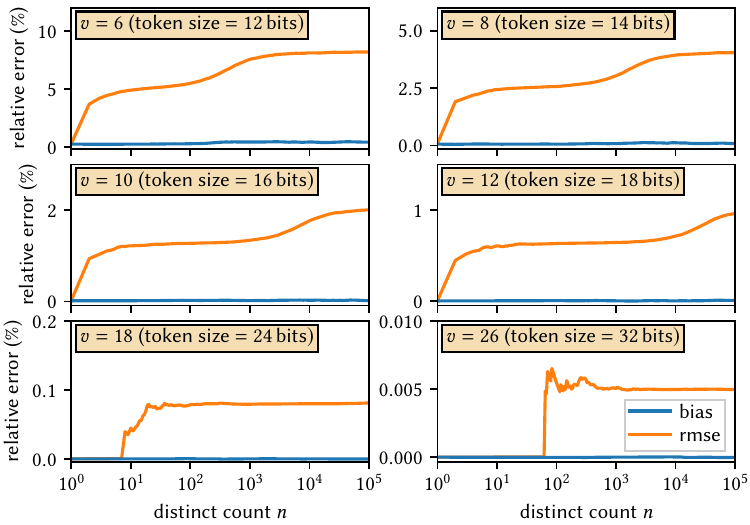}
	\caption{\boldmath The relative bias and the \acs*{RMSE} when estimating the distinct count from a set of collected distinct hash tokens with different sizes.}
	\label{fig:token-estimation-error}
\end{figure}

We provide a Java reference implementation of \acf{ELL} together with instructions and source code to reproduce all the presented results and figures at \url{https://github.com/dynatrace-research/exaloglog-paper}. The repository also includes numerous unit tests that cover \qty{100}{\percent} of the code. In particular, merging of \ac{ELL} sketches based on \Cref{alg:merge_register} was tested by creating many pairs of random \ac{ELL} sketches for which we compared the merged \ac{ELL} sketch with a sketch into which the unified stream of elements was inserted. Similarly, \Cref{alg:reduction} was tested by inserting the identical elements into two \ac{ELL} sketches with different configurations and checking whether the state was the same after reduction to the same parameters.

\subsection{Estimation Error}

We need to verify that the estimation error of \ac{ELL} really matches the theoretically predicted error, despite modeling \eqref{equ:step_dist} by \eqref{equ:geometric} and the number of distinct element insertions by a Poisson distribution. According to \eqref{equ:mvp_def} we would expect a \ac{RMSE} error of $\sqrt{\symMVP / ((\symBitsForMax + \symNumExtraBits)\symNumReg )}$ with $\symBitsForMax=6+\symExtraBits$, $\symBase=2^{2^{-\symExtraBits}}$, and the $\symMVP$ either given by \eqref{equ:mvp_uncompressed} or \eqref{equ:mvp_martingale}, depending on whether the \ac{ML} or the martingale estimator is used.

An accurate evaluation of the error requires thousands of estimates of different data sets with identical true distinct count $\symCardinality$. It is infeasible to use real data sets, if we want to repeat that for many different and also large $\symCardinality$. Therefore, we use a more efficient approach \cite{Ertl2024}. Extensive empirical tests \cite{Urban} have shown that the output of modern hash functions such as WyHash \cite{Yi}, Komihash \cite{Vaneev}, or PolymurHash \cite{Peters} can be considered like uniform random values. Otherwise, field-tested probabilistic data structures like \ac{HLL} would not work. This fact allows us to perform the experiments without real or artificially generated data.

Insertion of a new element can be simulated by simply generating a 64-bit random value to be used directly as the hash value of the inserted element in \Cref{alg:insertion}. Duplicate insertions of the same element can be ignored as they cannot change the state of \ac{ELL} by definition. Processing a random data set with true distinct count $\symCardinality$ is thus equivalent to using $\symCardinality$ random values instead. Accidental collisions of random values can be ignored because they occur with the same probability as hash collisions for real data.
To simulate the estimation error for a predefined distinct count value $\symCardinality$, the estimate is computed after updating the \ac{ELL} sketch using \Cref{alg:insertion} with $\symCardinality$ random values and finally compared against the true distinct count $\symCardinality$. By repeating this process with many different random sequences, in our experiments \num{100000}, the bias and the \ac{RMSE} can be empirically determined.

As this approach becomes computationally infeasible for distinct counts beyond 1 million, we switch to a different strategy \cite{Ertl2024}. After the first million insertions, for which a random value was generated each time, we just generate the waiting time (the number of distinct count increments) until a register is processed with a certain update value the next time. The probability that a register is updated with any possible update value $\symUpdateVal\in[1, (65 - \symPrecision - \symExtraBits)2^\symExtraBits]$ is given by $\symDensityUpdate(\symUpdateVal)/\symNumReg$ with $\symDensityUpdate(\symUpdateVal)$ from \eqref{equ:update_density}.
Therefore, the number of distinct count increments until a register is updated with a specific value $\symUpdateVal$ the next time is geometrically distributed with corresponding success probability.
In this way, we determine the next update time for each register and for each possible update value. Since the same update value can only modify a register once, we do not need to consider further updates which might occur with the same value for the same register. Knowing these $\symNumReg \times ((65 - \symPrecision - \symExtraBits)2^\symExtraBits)$ distinct counts leading to possible state changes in advance, enables us to make large distinct count increments, resulting in a huge speedup. This eventually allowed us to simulate the estimation error for distinct counts up to values of \num{e21} and also to test the presented estimators up to the exa-scale.

\Cref{fig:estimation-error} shows the empirical relative bias and \ac{RMSE} together with the theoretical \ac{RMSE} given by $\sqrt{\symMVP / ((\symBitsForMax + \symNumExtraBits)\symNumReg )}$ for the \ac{ML} and the martingale estimator for configurations $(\symExtraBits, \symNumExtraBits)\in\lbrace(1,9),(2,16),(2,20),(2,24) \rbrace$ and precisions $\symPrecision\in\lbrace 4, 6, 8, 10\rbrace$.
For intermediate distinct counts, perfect agreement with theory is observed.
For small distinct counts the estimation error is even much smaller. Interestingly, the estimation error also decreases slightly at the end of the operating range, which is in the order of $2^{64} \approx \num{1.9e19}$ and thus lies in the exa-scale range.
The estimators are essentially unbiased. The tiny bias which appears for small $\symPrecision$ for the \ac{ML} estimator can be ignored in practice as it is much smaller than the theoretical \ac{RMSE}.

We also verified estimation from sets of hash tokens as proposed in \Cref{sec:sparse_mode}. Again, we performed \num{100000} simulation runs, where we simulated 64-bit hash values by taking 64-bit random values and transforming them into corresponding hash tokens using different parameters $\symHashTokenParameter\in\lbrace 6,8,10,12,18,26\rbrace$. We considered distinct counts up to \num{e5}, which is typically far beyond the break-even point where a transition to the dense representation takes place. As shown in \Cref{fig:token-estimation-error}, the estimates are unbiased, and the relative estimation error is slightly smaller than the estimation error of an \ac{ELL} sketch for which $\symPrecision+\symExtraBits = \symHashTokenParameter$. The reason is that the set of hash tokens contains information that is equivalent to an \ac{ELL} sketch with $\symNumExtraBits\rightarrow\infty$. From a practical perspective, a hash token size of 4 bytes ($\symHashTokenParameter=26$) is particularly interesting, because it is big enough to support any practical \ac{ELL} configurations. Furthermore, as the tokens can be stored in a plain 32-bit integer array, off-the-shelve sorting algorithms can be used for deduplication.

\subsection{Space Efficiency Comparison}

\begin{table*}[t]
	\begin{threeparttable}
		\caption{\boldmath Comparison of mergeable approximate distinct counting algorithms when estimating the distinct count after inserting $\symCardinality = \num{e6}$ distinct elements. The parameters were chosen to obtain roughly \qty{2}{\percent} \acf*{RMSE}. The actual \acs*{RMSE} was empirically determined from 1 million simulation runs. The \acf*{MVP} is estimated as $\symMVP=(\text{average memory/serialization size in bits})\times (\text{\acs*{RMSE}})^2$ and is a fair measure of the space efficiency. The table is sorted by the in-memory \acs*{MVP}.}
		\label{tab:comparison}

		\scriptsize
		\begin{tabular*}{\linewidth}{@{\extracolsep{\fill}} llllrrrrr}
			\toprule
			&
			&
			&
			&
			\multicolumn{2}{c}{Size in bytes}
			&
			\multicolumn{2}{c}{\acs*{MVP} (space efficiency)}
			&
			\multirow[b]{2}{*}{\makecell{Constant-time\\insert operation}}
			\\
			\cmidrule(lr){5-6}
			\cmidrule(lr){7-8}
			Algorithm
			&
			References
			&
			Source code (https://github.com/...)
			&
			\acs*{RMSE}
			&
			memory
			&
			serialized
			&
			memory
			&
			serialized
			&
			\\
			\midrule
			\acl*{HLL} (\acs*{HLL}, 8-bit registers, $\symPrecision=11$) & \cite{ApacheDataSketches} & \href{https://github.com/apache/datasketches-java}{apache/datasketches-java} & \qty{2.29}{\percent} & \num{2296} & \num{2088} & \num{9.66} & \num{8.78}& \checkmark \\
			\acl*{HLL} (\acs*{HLL}, 6-bit registers, $\symPrecision=11$) & \cite{ApacheDataSketches,Heule2013} & \href{https://github.com/apache/datasketches-java}{apache/datasketches-java} & \qty{2.29}{\percent} & \num{1792} & \num{1577} & \num{7.54} & \num{6.63} & \checkmark \\
			\acl*{HLL} (\acs*{HLL}, \acs*{ML} estimator, $\symPrecision=11$) & \cite{Ertl2017} & \href{https://github.com/dynatrace-oss/hash4j}{dynatrace-oss/hash4j} & \qty{2.29}{\percent} & \num{1576} & \num{1536} & \num{6.63} & \num{6.46} & \checkmark \\
			\acl*{HLL} (\acs*{HLL}, 4-bit registers, $\symPrecision=11$) & \cite{ApacheDataSketches} & \href{https://github.com/apache/datasketches-java}{apache/datasketches-java} & \qty{2.29}{\percent} & \num[separate-uncertainty=true]{1331\pm 56} & \num[separate-uncertainty=true]{1067\pm 4} & \num{5.60} & \num{4.49} & -- \\
			\Acl*{CPC} (\acs*{CPC}, $\symPrecision=10$) & \cite{ApacheDataSketches,Lang2017} & \href{https://github.com/apache/datasketches-java}{apache/datasketches-java} & \qty{2.16}{\percent} & \num[separate-uncertainty=true]{1416\pm 34} & \num[separate-uncertainty=true]{656\pm 11} & \num{5.30} & \num{2.46}\TPTrlap{\tnote{\textasteriskcentered}} & -- \\
			\acl*{ULL} (\acs*{ULL}, \acs*{ML} estimator, $\symPrecision=10$) & \cite{Ertl2024} & \href{https://github.com/dynatrace-oss/hash4j}{dynatrace-oss/hash4j} & \qty{2.38}{\percent} & \num{1056} & \num{1024} & \num{4.78} & \num{4.64} & \checkmark \\
			\acl*{HLLL} (\acs*{HLLL}, $\symPrecision=11$) & \cite{Karppa2022} & \href{https://github.com/mkarppa/hyperlogloglog}{mkarppa/hyperlogloglog} & \qty{2.30}{\percent} & \num[separate-uncertainty=true]{1100\pm 13} & \num[separate-uncertainty=true]{898\pm 16} & \num{4.64} & \num{3.79} & -- \\
			SpikeSketch (128 buckets) & \cite{Du2023} & \href{https://github.com/duyang92/SpikeSketch}{duyang92/SpikeSketch} & \qty{2.26}{\percent}\TPTrlap{\tnote{\textasteriskcentered\textasteriskcentered}} & $\geq\num{1024}$\TPTrlap{\tnote{\textasteriskcentered\textasteriskcentered\textasteriskcentered}} & $\geq\num{1024}$\TPTrlap{\tnote{\textasteriskcentered\textasteriskcentered\textasteriskcentered}} & $\geq\num{4.19}$\TPTrlap{\tnote{\textasteriskcentered\textasteriskcentered\textasteriskcentered}} & $\geq\num{4.19}$\TPTrlap{\tnote{\textasteriskcentered\textasteriskcentered\textasteriskcentered}} & \checkmark \\
			\acl*{ELL} (\acs*{ELL}, $\symExtraBits=2$, $\symNumExtraBits=24$, $\symPrecision=8$) & this work & \href{https://github.com/dynatrace-research/exaloglog-paper}{dynatrace-research/exaloglog-paper} & \qty{2.15}{\percent} & \num{1064} & \num{1024} & \num{3.93} & \num{3.79} & \checkmark \\
			\acl*{ELL} (\acs*{ELL}, $\symExtraBits=2$, $\symNumExtraBits=20$, $\symPrecision=8$) & this work & \href{https://github.com/dynatrace-research/exaloglog-paper}{dynatrace-research/exaloglog-paper} & \qty{2.27}{\percent} & \num{936} & \num{896} & \num{3.86} & \num{3.69} & \checkmark \\
			\midrule
			Conjectured lower bound & \cite{Pettie2021} & -- & -- & -- & -- & \num{1.98} & \num{1.98} & not known \\
			\bottomrule
		\end{tabular*}
		\begin{tablenotes}
			\item [\textasteriskcentered] achieved by expensive compression during serialization
			\item [\textasteriskcentered\textasteriskcentered] error can be much larger for smaller distinct counts
			\item [\textasteriskcentered\textasteriskcentered\textasteriskcentered] lower bound values based on theoretical considerations completely ignoring auxiliary data fields (empirical values are meaningless as the reference implementation is not optimized)
		\end{tablenotes}
	\end{threeparttable}
\end{table*}

Our experiments have shown that the estimation error matches the theoretical predicted estimation error. Since the space requirement of \ac{ELL} is constant $(\symBitsForMax + \symNumExtraBits)\symNumReg$ bits, the theoretically predicted \ac{MVP}, as discussed in \Cref{sec:choice_param}, can be achieved if memory overhead for the Java object or auxiliary fields can be ignored. However, for a fairer comparison with other practical algorithms, we considered the empirical \ac{MVP} based on the total space allocated by the whole data structure.

We performed 1 million simulation runs. In each cycle, the distinct count was estimated and the allocated amount of memory as well as the serialization size were measured after adding up to \num{e6} distinct random elements. This allowed us to compute the \ac{RMSE} and, together with the average space requirements, the empirical \acp{MVP} according to \eqref{equ:mvp_def}.

\Cref{tab:comparison} compares our \ac{ELL} reference implementation to other state-of-the-art algorithms. All algorithms were configured to give roughly \qty{2}{\percent} estimation error. Since the reference implementation of SpikeSketch \cite{Du2023} is not very space-efficient and also does not support serialization, we used the size of the plain register array without any additional overheads as lower bound size estimates. For algorithms that allocate variable space, such as \ac{HLL} with 4-bit registers, \ac{CPC}, or \ac{HLLL}, the standard deviation of the size is also shown.

The serialization size is always smaller than the in-memory size to which object overhead or auxiliary fields such as buffers also contribute. The difference is particularly large for \ac{CPC}, whose serialization method also applies a specialized and relatively expensive (see \Cref{sec:performance}) compression step \cite{Lang2017}. A fair comparison would require the development of specific compression techniques for all other data structures which is out of the scope of this work. However, the theoretical \acp{MVP} for optimal compression shown in \Cref{fig:mvp_ml_compressed} indicate that the size of \ac{ELL} could be further reduced. For \ac{ULL}, which is a special case of \ac{ELL}, we have already shown that \acp{MVP} below 3 can be achieved. Its 1-byte register array seems to be very convenient for standard compression algorithms \cite{Ertl2024}.

\Cref{fig:memory} also shows the memory consumption and the corresponding \ac{MVP} for other distinct counts. \ac{ELL} requires constant space and never allocates additional data. The data structures from the Apache DataSketches library have implemented a sparse mode that allows them to be more space-efficient for small distinct counts. However, a sparse mode could also be easily implemented for \ac{ELL} as discussed in \Cref{sec:sparse_mode}.

The \ac{MVP} of SpikeSketch stands out at lower distinct counts. The high values are a result of the lossy compression and stepwise smoothing. The latter reduces the update probability even of empty SpikeSketches by a factor of \qty{64}{\percent}. As a consequence, the estimation error is \qty{100}{\percent} with a \qty{36}{\percent} probability independent of the number of buckets for the extreme case of $\symCardinality=1$. We, therefore, do not consider SpikeSketch to be suitable for practical use. \ac{HLLL} also shows a spike around $\symCardinality=\num{5e3}$, which is a result of using the original \ac{HLL} estimator \cite{Flajolet2007} that is known to have some issues \cite{Heule2013,Ertl2017}.

\subsection{Performance Comparison}
\label{sec:performance}

To compare the performance of \ac{ELL} with the other algorithms listed in \Cref{tab:comparison}, we used an Amazon EC2 c5.metal instance running Ubuntu Server 24.04 LTS. Turbo Boost was disabled by setting the processor P-state to 1 \cite{AmazonPState}. We used Java implementations for all algorithms except for \ac{HLLL} and SpikeSketch whose reference implementations are written in C++.

\Cref{fig:benchmarks} shows the results of our benchmarks. First, element insertion was tested by adding up to \num{e6} random 16-byte arrays that were generated and stored in memory in advance.
As Apache DataSketches uses the 128-bit version of Murmur3 as built-in hash function without the flexibility of defining a different one, we used it also for all other algorithms to make a fair comparison. The corresponding graph shows the average time per inserted element, which also includes the initial allocation of the data structure. For this reason, the measured times for small $\symCardinality$ tend to be higher. All insertion times, except those for \ac{HLLL} and SpikeSketch are between \num{20} and \qty{50}{\nano\second}. For \ac{ELL}, we investigated insertion with and without martingale estimator.

The fastest estimation times are achieved by algorithms, including those from Apache DataSketches, that maintain a martingale estimator or keep track of other redundant statistics during insertion. The \ac{ELL} \ac{ML} estimator compares well to those algorithms that insert elements without such additional bookkeeping.

To analyze serialization, we measured the time to write the state into a newly allocated byte array. For \ac{ELL}, this means just copying the byte array holding the register values, which is very fast. The results also show that the serialization of \ac{CPC} is more than an order of magnitude slower due to the expensive compression, as discussed before. We do not have data for \ac{HLLL} and SpikeSketch as their reference implementations come without a serialization method.

Finally, we also investigated the time to merge two data structures, both filled with $\symCardinality$ random elements. The results show that \ac{ELL} is very fast. Algorithms that have implemented a sparse mode are, as expected, faster for small $\symCardinality$ as less data has to be processed. The comparison with the algorithms from Apache DataSketches is not entirely fair, as they rebuild internal statistics for estimation during the merging process. For this reason, we have also considered merging followed by estimation, which is a common operation sequence in practice. \ac{ELL} also performs quite well in this case. We have no data for SpikeSketch as its reference implementation does not include a working merge operation.

It is worth noting that our \ac{ELL} reference implementation is generic and supports arbitrary values of $\symExtraBits$ and $\symNumExtraBits$. Hardcoding these values could potentially further improve its performance.

\section{Future Work}
A topic for future research is the compressibility of \ac{ELL}. According to \Cref{fig:mvp_ml_compressed,fig:mvp_martingale_compressed}, much lower \acp{MVP} could be achieved with optimal compression. For \ac{ULL}, a special case of \ac{ELL} with $\symExtraBits=0$ and $\symNumExtraBits=2$ (compare \Cref{sec:related_data_structures}), a reduction close to the theoretical limit can be achieved with standard compression algorithms \cite{Ertl2024} such as Zstandard \cite{Collet2018}. Unlike \ac{ULL}, whose register size is exactly one byte, we assume that standard algorithms will work worse in the general case. Since the shape of the register distribution is known (see \Cref{sec:pmf_reg}), some sort of entropy coding could be a way to approach the theoretical limit.

As discussed in \Cref{sec:related_data_structures}, HyperMinHash is a special case of \ac{ELL}, and \ac{PCSA} and CPC contain the same information as an \ac{ELL}(0, 64) sketch.
Therefore, our proposed \ac{ML} estimation approach, in which the \ac{ML} equation is first reduced to the simple form \eqref{equ:log_likelihood_shape} with a relatively small number of terms, should also work for them. We assume that this method could lead to slightly lower estimation errors than current approaches, as \ac{ML} estimation is generally known to be asymptotically efficient.

\begin{figure}[t]
	\centering
	\includegraphics[width=1\columnwidth]{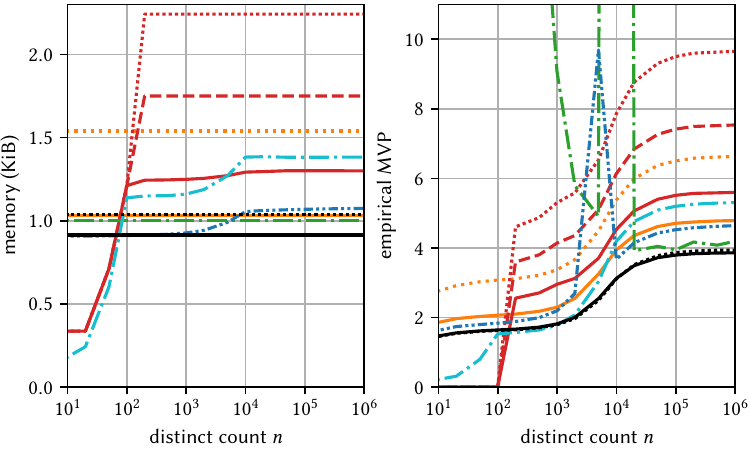}
	\caption{\boldmath The average memory footprint and the empirical \acs*{MVP} for $\symCardinality\in\lbrace 10,20,50,100,200,500,\ldots,\num{e6}\rbrace$ obtained from 1 million simulation runs. The legend is given in \Cref{fig:benchmarks}.}
	\label{fig:memory}
\end{figure}

\begin{figure}[t]
	\centering
	\includegraphics[width=1\columnwidth]{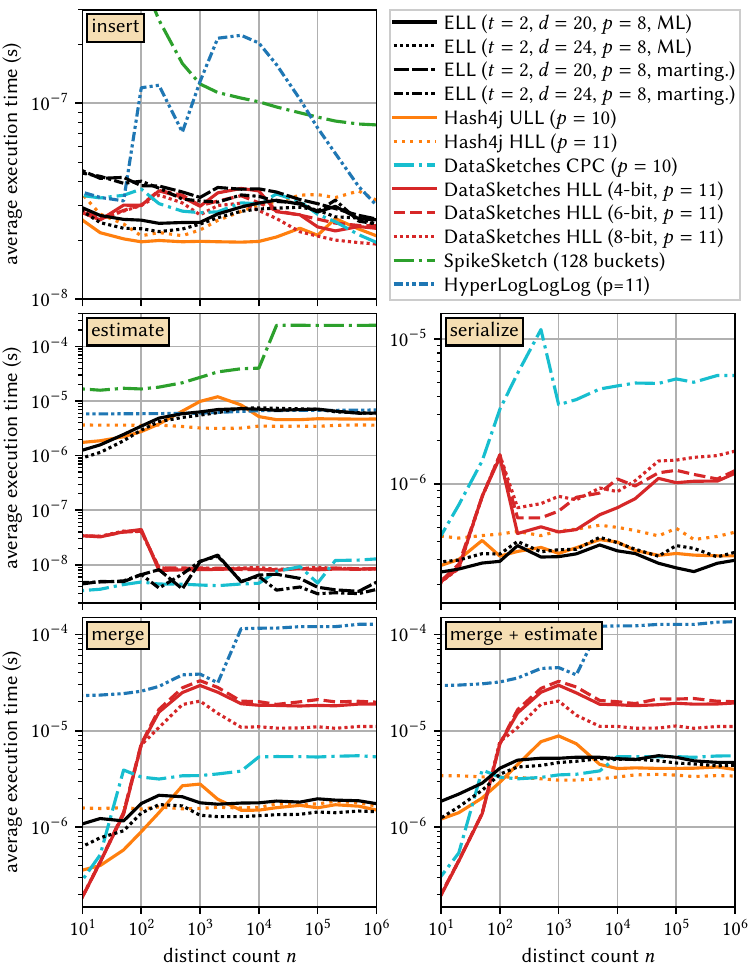}
	\caption{\boldmath The average execution time for insert, estimate, serialize, merge, and combined merge and estimate operations for $\symCardinality\in\lbrace 10,20,50,100,200,500,\ldots,\num{e6}\rbrace$.}
	\label{fig:benchmarks}
\end{figure}

\section{Conclusion}
We have introduced a new algorithm for distinct counting called \acf{ELL}, which includes already-known algorithms as special cases. With the right parameters, the space efficiency can be improved significantly. In particular, a configuration was presented that reduces the \ac{MVP} by \qty{43}{\percent} compared to the widely used HyperLogLog algorithm. \ac{ELL} also supports practical properties such as mergeability, idempotency, reproducibility, and reducibility. We have also shown that \acf{ML} estimation using the developed robust numerical solver and martingale estimation are feasible. The observed estimation errors are fully consistent with the theory. In contrast to other recent approaches, insertions always require constant time, regardless of the chosen accuracy. We also proposed a sparse mode based on hash tokens, which allows the allocation of the register array to be postponed. All this makes \ac{ELL} very attractive for wider use in practice.

\appendix

\section{Numerical Root-Finding}
\label{app:numerical-root-finding}

As the function $\symFunc$ defined in \eqref{equ:ml_func} is strictly increasing and concave (see \Cref{lem:inc_and_concave}), the root of the \ac{ML} equation $\symFunc(\symX)=0$ can be robustly found using Newton's method when starting from a point $\symX_{0}$ satisfying $\symFunc(\symX_{0})\leq 0$. According to \Cref{lem:inequality} such a point is
\begin{multline}
	\label{equ:starting_point}
	\textstyle
	\symX_{0}
	=
	\exp(
	\ln(
		1
		+
		\frac{
			\symLikelihoodFuncExponentTwoPow
		}
		{
			\symLikelihoodFuncExponentOne 2^{\symMaxUpdateValMax}
		}
		)
	\frac
		{
			\symLikelihoodFuncExponentTwoSum
		}
		{
			\symLikelihoodFuncExponentTwoPow	}
	)
	-1
	\\
	\text{with $\textstyle\symLikelihoodFuncExponentTwoSum:=\sum_{\symIndexJ=\symMaxUpdateValMin}^{\symMaxUpdateValMax}
			\symLikelihoodFuncExponentTwo_{\symIndexJ}$ and $\textstyle
			\symLikelihoodFuncExponentTwoPow:=\sum_{\symIndexJ=\symMaxUpdateValMin}^{\symMaxUpdateValMax}
			\symLikelihoodFuncExponentTwo_{\symIndexJ}
			2^{\symMaxUpdateValMax - \symIndexJ}$}.
\end{multline}
As a result, the sequence of points obtained by Newton's method
\begin{equation*}
	\textstyle
	\symX_{\symTimeIndex+1}
	=
	\symX_{\symTimeIndex}
	-
	\frac{\symFunc(\symX_{\symTimeIndex})}{\symFunc'(\symX_{\symTimeIndex})}
	=
	\symX_{\symTimeIndex}
	\left(
	1
	-
	\frac{\symFunc(\symX_{\symTimeIndex})}{\symX_{\symTimeIndex}\symFunc'(\symX_{\symTimeIndex})}
	\right)
\end{equation*}
is always increasing and approaches $\symXZero$. The recursion can be transformed using
\begin{multline*}
	\textstyle
	\symX \symFunc'(\symX)-\symLikelihoodFuncExponentOne 2^{\symMaxUpdateValMax}\symX
	=
	-\symX\frac{\partial}{\partial \symX}\sum_{\symIndexJ=1}^{\symMaxUpdateValMax - \symMaxUpdateValMin}\symLikelihoodFuncExponentTwo_{\symMaxUpdateValMax - \symIndexJ}
	\frac{ 2^\symIndexJ \symX}{(1+\symX)^{2^\symIndexJ}-1}
	\\
	\begin{aligned}
		 & =
		\textstyle
		\sum_{\symIndexJ=1}^{\symMaxUpdateValMax - \symMaxUpdateValMin}\symLikelihoodFuncExponentTwo_{\symMaxUpdateValMax - \symIndexJ}
		\left(
		\frac{4^\symIndexJ \symX^2 (1 + \symX)^{2^\symIndexJ-1} }{((1 + \symX)^{2^\symIndexJ}-1)^2}
		-
		\frac{ 2^\symIndexJ \symX}{(1+\symX)^{2^\symIndexJ}-1}
		\right)
		\\
		 & =
		\textstyle
		\sum_{\symIndexJ=1}^{\symMaxUpdateValMax - \symMaxUpdateValMin}\symLikelihoodFuncExponentTwo_{\symMaxUpdateValMax - \symIndexJ}
		\frac{ 2^\symIndexJ \symX}{(1+\symX)^{2^\symIndexJ}-1}
		\left(
		\frac{2^\symIndexJ \symX (1 + \symX)^{2^\symIndexJ-1} }{(1 + \symX)^{2^\symIndexJ}-1}
		-
		1
		\right)
		=
		\symSumB(\symX)
	\end{aligned}
\end{multline*}
with
\begin{multline}
	\label{equ:sumb}
	\textstyle
	\symSumB(\symX)
	:=
	\sum_{\symIndexJ=1}^{\symMaxUpdateValMax - \symMaxUpdateValMin}
	\symLikelihoodFuncExponentTwo_{\symMaxUpdateValMax - \symIndexJ}
	\left(
	\prod_{\symIndexL=0}^{\symIndexJ-1}
	\frac{ 2}{(1+\symX)^{2^\symIndexL}+1}
	\right)
	\cdot
	\\
	\textstyle
	\cdot
	\left(
	\left(
		\prod_{\symIndexL = 0}^{\symIndexJ-1}
		\frac{2 (1 + \symX)^{2^\symIndexL} }{(1 + \symX)^{2^\symIndexL}+1}
		\right)
	-1
	\right)
\end{multline}
into
\begin{equation}
	\label{equ:newton_iteration}
	\textstyle
	\symX_{\symTimeIndex+1}
	=
	\symX_{\symTimeIndex}
	\left(
	1
	+
	\frac{
		\symSumA(\symX_{\symTimeIndex}) -\symLikelihoodFuncExponentOne 2^{\symMaxUpdateValMax}\symX_{\symTimeIndex}
	}{
		\symSumB(\symX_{\symTimeIndex})
		+
		\symLikelihoodFuncExponentOne 2^{\symMaxUpdateValMax}\symX_{\symTimeIndex}
	}
	\right)
	.
\end{equation}
Function $\symSumB$ shares the same product as $\symSumA$ \eqref{equ:def_suma}, which therefore needs to be computed only once. The last factor in \eqref{equ:sumb} is always nonnegative as the product consists of factors that are all greater than or equal to 1, because
$1\leq \frac{2 (1 + \symX)^{2^\symIndexL} }{(1 + \symX)^{2^\symIndexL}+1} < 2$.
Furthermore, when using standard floating-point types, overflows will not occur as the product has at most $\symMaxUpdateValMax - \symMaxUpdateValMin\leq 63-\symPrecision-\symExtraBits$ factors and all of them are smaller than 2.
To reduce numerical errors and minimize the number of operations, we compute
\begin{equation*}
	\textstyle
	\symFactorTwo_\symIndexJ:=
	\left(
	\prod_{\symIndexL = 0}^{\symIndexJ-1}
	\frac{2 (1 + \symX)^{2^\symIndexL} }{(1 + \symX)^{2^\symIndexL}+1}
	\right)
	-1
\end{equation*}
recursively using \eqref{equ:y_def} according to
\begin{equation}
	\label{equ:factor_two_def}
	\textstyle
	\symFactorTwo_{\symIndexJ+1}
	=
	\symFactorTwo_{\symIndexJ} (2 - \frac{2}{2+\symY_{\symIndexJ}}) + (1- \frac{2}{2+\symY_{\symIndexJ}}).
\end{equation}

The Newton iteration \eqref{equ:newton_iteration} can be stopped, if $\symFunc(\symX_{\symTimeIndex})\geq 0$, equivalent to $\symSumA(\symX_{\symTimeIndex}) \leq \symLikelihoodFuncExponentOne 2^{\symMaxUpdateValMax}\symX_{\symTimeIndex}$, because we expect the sequence $\symX_{\symTimeIndex}$ to be increasing and converging towards the root. The case $\symSumA(\symX_{\symTimeIndex}) < \symLikelihoodFuncExponentOne 2^{\symMaxUpdateValMax}\symX_{\symTimeIndex}$ may happen due to numerical errors. It is reasonable to stop the Newton iteration also in this case, as the numerical error limits have been reached. Similarly, we stop the iteration, if $\symX_{\symTimeIndex+1}\leq \symX_{\symTimeIndex}$.
In practice, only a small number of iterations is needed to satisfy any of the two stop conditions. In all our experiments presented in \Cref{sec:experiments}, the number of iterations never exceeded 10 when calculating the estimate from \ac{ELL} sketches. On average, we observed between 5 and 7 iterations, dependent on the \ac{ELL} parameters and the true distinct count.
The whole procedure to compute the \ac{ML} estimate using Newton's method is summarized by \Cref{alg:numerical_maximizer}.

\myAlg{
	\caption{Numerical computation of the distinct count estimate by solving the \acs*{ML} equation using Newton's method. $\symLikelihoodFuncExponentOne$ and $\symLikelihoodFuncExponentTwo_{\symExtraBits + 1}, \ldots, \symLikelihoodFuncExponentTwo_{64-\symPrecision}$ are the coefficients of the log-likelihood function \eqref{equ:log_likelihood_shape}.}
	\label{alg:numerical_maximizer}
	$\symLikelihoodFuncExponentTwoSum \gets 0 $, $\symLikelihoodFuncExponentTwoPow \gets 0 $, $\symMaxUpdateValMin\gets -1$, $\symMaxUpdateValMax\gets 0$\;
	\For{$\symIndexJ \gets \symExtraBits + 1\ \KwTo \ 64 - \symPrecision$} {
		\If{$\symLikelihoodFuncExponentTwo_\symIndexJ > 0$}{
			\lIf{$\symMaxUpdateValMin < 0$}{$\symMaxUpdateValMin \gets \symIndexJ$}
			$\symMaxUpdateValMax \gets \symIndexJ$\;
			$\symLikelihoodFuncExponentTwoSum \gets\symLikelihoodFuncExponentTwoSum + \symLikelihoodFuncExponentTwo_\symIndexJ$, $\symLikelihoodFuncExponentTwoPow \gets\symLikelihoodFuncExponentTwoPow + \symLikelihoodFuncExponentTwo_\symIndexJ\cdot 2^{-\symIndexJ}$\Comment*[r]{see \eqref{equ:starting_point}}
		}
	}
	\lIf(\Comment*[f]{all $\symLikelihoodFuncExponentTwo_\symIndexJ$ are zero}){$\symMaxUpdateValMin < 0$}{\KwRet $0$}
	$\symLikelihoodFuncExponentTwoPow \gets \symLikelihoodFuncExponentTwoPow \cdot 2^\symMaxUpdateValMax$\;
	$\symX\gets \symLikelihoodFuncExponentTwoPow/(\symLikelihoodFuncExponentOne\cdot 2^{\symMaxUpdateValMax})$\;
	\If(\Comment*[f]{$\symX$ is already the root of $\symFunc$, if $\symMaxUpdateValMin = \symMaxUpdateValMax$}){$\symMaxUpdateValMin < \symMaxUpdateValMax$} {
		$\symX\gets \symExpMinusOne(\symLogPlusOne(\symX)
			\cdot (\symLikelihoodFuncExponentTwoSum/\symLikelihoodFuncExponentTwoPow
				))
		$\Comment*[r]{starting point, see \eqref{equ:starting_point}}
		\Loop(\Comment*[f]{main loop of Newton iteration}){} {
			$\symFactorOne\gets 1$, $\symFactorTwo\gets 0$, $\symY\gets\symX$, $\symMaxUpdateVal\gets\symMaxUpdateValMax$\;
			$\symSumA \gets \symLikelihoodFuncExponentTwo_\symMaxUpdateVal$, $\symSumB \gets 0$\;
			\Loop(\Comment*[f]{loop for summing up $\symSumA$ \eqref{equ:def_suma} and $\symSumB$ \eqref{equ:sumb}}){}{
				$\symMaxUpdateVal\gets\symMaxUpdateVal-1$\;
				$\symZ \gets 2/(2+\symY)$\Comment*[r]{$\symZ\in[0,1]$}
				$\symFactorOne\gets \symFactorOne \cdot \symZ$\Comment*[r]{$\symFactorOne$ is decreasing, compare \eqref{equ:factor_one_def}}
				$\symFactorTwo\gets \symFactorTwo \cdot (2 - \symZ) + (1 - \symZ) $\Comment*[r]{$\symFactorTwo$ is increasing, compare \eqref{equ:factor_two_def}}
				$\symSumA \gets \symSumA + \symLikelihoodFuncExponentTwo_\symMaxUpdateVal\cdot \symFactorOne$,
				$\symSumB \gets \symSumB + \symLikelihoodFuncExponentTwo_\symMaxUpdateVal\cdot \symFactorOne\cdot \symFactorTwo$\;
				\lIf{$\symMaxUpdateVal \leq \symMaxUpdateValMin$}{\Break}
				$\symY\gets\symY\cdot(\symY+2)$\Comment*[r]{compare \eqref{equ:y_recursion}}
			}
			$\symX'\gets(\symLikelihoodFuncExponentOne\cdot 2^{\symMaxUpdateValMax})\cdot \symX$\;
			\lIf(\Comment*[f]{stop iteration if $\symFunc(\symX) \geq 0$, see \eqref{equ:ml_func}}){$\symSumA \leq \symX'$}{\Break}
			$\symXOld\gets\symX$\;
			$\symX\gets\symX\cdot (1 + (\symSumA -\symX') / (\symSumB + \symX'))$\Comment*[r]{compare \eqref{equ:newton_iteration}}
			\lIf(\Comment*[f]{stop if numerically converged}){$\symX \leq \symXOld$}{\Break}
		}
	}
	\Return $\symNumReg \cdot 2^{\symMaxUpdateValMax}\cdot\symLogPlusOne(\symX)$ \Comment*[r]{compare \eqref{equ:final_estimate}}
}

\section{Proofs}
\label{app:proofs}

\begin{lemma}
	\label{lem:identity}
	For {$\symDensityUpdate$} and $\symExpFunc$ as defined in \eqref{equ:update_density} and \eqref{equ:exponent_func} $
		\sum_{\symUpdateVal = \symMaxUpdateVal+1}^{(65 - \symPrecision - \symExtraBits)2^\symExtraBits}\symDensityUpdate(\symUpdateVal)
		=
		\frac{2^{\symExtraBits}(1-\symExtraBits +\symExpFunc(\symMaxUpdateVal)) - \symMaxUpdateVal }{2^{\symExpFunc(\symMaxUpdateVal)}}
	$ holds.
\end{lemma}
\begin{proof}
	This formula can be proven by induction. For $\symMaxUpdateVal = (65 - \symPrecision - \symExtraBits)2^\symExtraBits$ both sides are zero.
	If the identity holds for $\symMaxUpdateVal$, it can be shown that it also holds for $\symMaxUpdateVal-1$. Since $\symExpFunc(\symMaxUpdateVal)-\symExpFunc(\symMaxUpdateVal-1)\in\lbrace 0,1\rbrace$, we show the identity for both possible cases. First, if $\symExpFunc(\symMaxUpdateVal)=\symExpFunc(\symMaxUpdateVal-1)$:
	\begin{multline*}
		\textstyle
		\sum_{\symUpdateVal = (\symMaxUpdateVal-1)+1}^{(65 - \symPrecision - \symExtraBits)2^\symExtraBits}\symDensityUpdate(\symUpdateVal)
		=
		\symDensityUpdate(\symMaxUpdateVal) + \sum_{\symUpdateVal = \symMaxUpdateVal+1}^{(65 - \symPrecision - \symExtraBits)2^\symExtraBits}\symDensityUpdate(\symUpdateVal)
		\\
		=
		\textstyle
		\frac{1}{2^{\symExpFunc(\symMaxUpdateVal)}}
		+
		\frac{2^{\symExtraBits}(1-\symExtraBits +\symExpFunc(\symMaxUpdateVal)) - \symMaxUpdateVal }{2^{\symExpFunc(\symMaxUpdateVal)}}
		=
		\frac{2^{\symExtraBits}(1-\symExtraBits +\symExpFunc(\symMaxUpdateVal-1)) - (\symMaxUpdateVal-1) }{2^{\symExpFunc(\symMaxUpdateVal-1)}}.
	\end{multline*}
	And, second, if $\symExpFunc(\symMaxUpdateVal)=\symExpFunc(\symMaxUpdateVal-1)+1$, which implies according to \eqref{equ:exponent_func} $\symMaxUpdateVal - 1 \equiv 0 \pmod{2^\symExtraBits}$ and further $\symExpFunc(\symMaxUpdateVal)=\symExtraBits+1+\frac{\symMaxUpdateVal-1}{2^\symExtraBits}\Leftrightarrow 2^\symExtraBits(\symExpFunc(\symMaxUpdateVal)-\symExtraBits-1)-\symMaxUpdateVal+1=0$. With the help of this identity we can show again
	\begin{multline*}
		\textstyle
		\sum_{\symUpdateVal = (\symMaxUpdateVal-1)+1}^{(65 - \symPrecision - \symExtraBits)2^\symExtraBits}\symDensityUpdate(\symUpdateVal)
		=
		\symDensityUpdate(\symMaxUpdateVal) + \sum_{\symUpdateVal = \symMaxUpdateVal+1}^{(65 - \symPrecision - \symExtraBits)2^\symExtraBits}\symDensityUpdate(\symUpdateVal)
		\\
		\begin{aligned}
			 & =
			\textstyle
			\frac{1}{2^{\symExpFunc(\symMaxUpdateVal)}}
			+
			\frac{2^{\symExtraBits}(1-\symExtraBits +\symExpFunc(\symMaxUpdateVal)) - \symMaxUpdateVal }{2^{\symExpFunc(\symMaxUpdateVal)}}
			+
			\frac{2^\symExtraBits(\symExpFunc(\symMaxUpdateVal)-\symExtraBits-1)-\symMaxUpdateVal+1}{2^{\symExpFunc(\symMaxUpdateVal)}}
			\\
			 & =
			\textstyle
			\frac{2^{\symExtraBits}(-2\symExtraBits +2\symExpFunc(\symMaxUpdateVal)) - 2(\symMaxUpdateVal-1) }{2^{\symExpFunc(\symMaxUpdateVal)}}
			=
			\frac{2^{\symExtraBits}(1-\symExtraBits +\symExpFunc(\symMaxUpdateVal-1)) - (\symMaxUpdateVal-1) }{2^{\symExpFunc(\symMaxUpdateVal-1)}}.
		\end{aligned}
	\end{multline*}
\end{proof}

\begin{lemma}
	\label{lem:inc_and_concave}
	The function { $\symFunc(\symX)$} as defined in \eqref{equ:ml_func} is strictly increasing and concave for {$\symX\geq 0$} and { $\symLikelihoodFuncExponentOne>0$}.
\end{lemma}
\begin{proof}
	It is sufficient to show that $\symOtherFunc(\symX):=-\frac{\symX}{(1+\symX)^\symSomeConstant-1}$ with $\symSomeConstant\in\mathbb{Z}^{+}$ is increasing and concave, which is the case if $\symOtherFunc'(\symX)\geq 0$ and $\symOtherFunc''(\symX)\leq 0$. The first derivative is given by
	\begin{align*}
		% https://www.wolframalpha.com/input?i=d%2Fdx+-x%2F%28%281%2Bx%29%5Ec-1%29
		\symOtherFunc'(\symX)
		 & =
		\textstyle\frac{\symSomeConstant \symX (1 + \symX)^{\symSomeConstant-1} - ((1 + \symX)^{\symSomeConstant}-1)}{((1 + \symX)^\symSomeConstant-1)^2}
		=
		\textstyle\frac{\symSomeConstant \symX (1 + \symX)^{\symSomeConstant-1} - \sum_{\symIndexJ=0}^{\symSomeConstant-1} \symX (1 + \symX)^{\symIndexJ}}{((1 + \symX)^\symSomeConstant-1)^2}
		\\
		 & =
		\textstyle\frac{\symX \sum_{\symIndexJ=0}^{\symSomeConstant-2}(1 + \symX)^{\symSomeConstant-1} - (1 + \symX)^{\symIndexJ}}{((1 + \symX)^\symSomeConstant-1)^2}
		\geq 0,
	\end{align*}
	which is nonnegative.
	The second derivative can be expressed as
	\begin{align*}
		% https://www.wolframalpha.com/input?i=d%5E2%2Fdx%5E2+-x%2F%28%281%2Bx%29%5Ec-1%29
		\symOtherFunc''(\symX)
		 & =
		\textstyle
		\frac{
			\symSomeConstant (1 + \symX)^{\symSomeConstant-2} \left((2 + \symX) ((1 + \symX)^\symSomeConstant-1) - \symSomeConstant \symX (1 + (1 + \symX)^{\symSomeConstant})\right)
		}{
			((1 + \symX)^{\symSomeConstant}-1)^3
		}
		\\
		 & =
		\textstyle
		\frac{
			\symSomeConstant (1 + \symX)^{\symSomeConstant-2} \left(2 ((1 + \symX)^\symSomeConstant-1) - (\symSomeConstant+1) \symX(1 + \symX)^0 - (\symSomeConstant-1) \symX (1 + \symX)^{\symSomeConstant}\right)
		}{
			((1 + \symX)^{\symSomeConstant}-1)^3
		}
		\\
		 & =
		\textstyle
		-\frac{
			\symSomeConstant (1 + \symX)^{\symSomeConstant-2} \symX
			\left((\symSomeConstant-1) (1 + \symX)^{\symSomeConstant} + (\symSomeConstant-1) (1 + \symX)^0
			-
			2\sum_{\symIndexJ=1}^{\symSomeConstant-1}
			(1 + \symX)^{\symIndexJ}
			\right)
		}{
			((1 + \symX)^{\symSomeConstant}-1)^3
		}.
	\end{align*}
	Since $(1+\symX)^\symY$ is convex with respect to $\symY$,
	\begin{multline*}
		\overbrace{(1+\symX)^{\symSomeConstant} + \ldots + (1+\symX)^{\symSomeConstant}}^{\text{$(\symSomeConstant-1)$ terms}}+\overbrace{(1+\symX)^{0}+\ldots+(1+\symX)^{0}}^{\text{$(\symSomeConstant-1)$ terms}}
		\\
		\geq
		\overbrace{(1+\symX)^{\symSomeConstant-1} + (1+\symX)^{\symSomeConstant-1} + \ldots + (1+\symX)^{1} + (1+\symX)^{1}}^{\text{$2(\symSomeConstant-1)$ terms}}
	\end{multline*}
	holds according to Karamata's inequality, which shows that the last factor of the numerator is nonnegative. Since all factors are nonnegative we have $\symOtherFunc''(\symX)\leq 0$.
\end{proof}

\begin{lemma}
	\label{lem:inequality}
	The root $\symXZero$ of $\symFunc$ as defined in \eqref{equ:ml_func} can be bracketed by
	$
		\exp(
		\ln\!\left(
			1
			+
			\frac{
				\symLikelihoodFuncExponentTwoPow
			}
			{
				\symLikelihoodFuncExponentOne 2^{\symMaxUpdateValMax}
			}
			\right)
		\frac
			{
				\symLikelihoodFuncExponentTwoSum
			}
			{
				\symLikelihoodFuncExponentTwoPow	}
		)
		-1
		\leq
		\symXZero
		\leq
		\frac{\symLikelihoodFuncExponentTwoSum}{\symLikelihoodFuncExponentOne 2^{\symMaxUpdateValMax}}
	$
	where {$\symLikelihoodFuncExponentTwoSum:=\sum_{\symIndexJ=\symMaxUpdateValMin}^{\symMaxUpdateValMax}
				\symLikelihoodFuncExponentTwo_{\symIndexJ}$} and {$\symLikelihoodFuncExponentTwoPow:=\sum_{\symIndexJ=\symMaxUpdateValMin}^{\symMaxUpdateValMax}
				\symLikelihoodFuncExponentTwo_{\symIndexJ}
				2^{\symMaxUpdateValMax - \symIndexJ}$}.
\end{lemma}
\begin{proof}
	We rewrite $\symFunc(\symXZero) = 0$ as
	\begin{equation}
		\label{equ:inequ_proof}
		\textstyle
		\symLikelihoodFuncExponentOne 2^{\symMaxUpdateValMax}\symXZero
		-
		\symLikelihoodFuncExponentTwoSum
		\frac{\sum_{\symIndexJ=0}^{\symMaxUpdateValMax - \symMaxUpdateValMin}
			\symLikelihoodFuncExponentTwo_{\symMaxUpdateValMax - \symIndexJ} \symOtherFunc(\symXZero, 2^\symIndexJ)}
		{
			\sum_{\symIndexJ=0}^{\symMaxUpdateValMax - \symMaxUpdateValMin}
			\symLikelihoodFuncExponentTwo_{\symMaxUpdateValMax - \symIndexJ}}
		=
		0
	\end{equation}
	with $\symOtherFunc(\symX, \symY) := \frac{\symX \symY }{(1+\symX)^{\symY}-1}$. As $\symOtherFunc$ is convex with respect to $\symY$, which follows from $\frac{\symZ}{e^\symZ-1}$ being convex with $\symZ = \ln(1+\symX)\symY$, we can apply Jensen's inequality
	\begin{multline*}
		\textstyle
		\symLikelihoodFuncExponentOne 2^{\symMaxUpdateValMax}\symXZero
		-
		\symLikelihoodFuncExponentTwoSum
		\cdot
		\symOtherFunc\!
		\left(
		\symXZero,
		\frac{\symLikelihoodFuncExponentTwoPow}
		{
			\symLikelihoodFuncExponentTwoSum
		}
		\right)
		\\
		\textstyle
		=
		\symLikelihoodFuncExponentOne 2^{\symMaxUpdateValMax}\symXZero
		-
		\symLikelihoodFuncExponentTwoSum
		\cdot
		\symOtherFunc\!
		\left(
		\symXZero,
		\frac{\sum_{\symIndexJ=0}^{\symMaxUpdateValMax - \symMaxUpdateValMin}
			\symLikelihoodFuncExponentTwo_{\symMaxUpdateValMax - \symIndexJ} 2^\symIndexJ}
		{
			\sum_{\symIndexJ=0}^{\symMaxUpdateValMax - \symMaxUpdateValMin}
			\symLikelihoodFuncExponentTwo_{\symMaxUpdateValMax - \symIndexJ}}
		\right)
		\\
		\textstyle
		\geq
		\symLikelihoodFuncExponentOne 2^{\symMaxUpdateValMax}\symXZero
		-
		\symLikelihoodFuncExponentTwoSum
		\frac{\sum_{\symIndexJ=0}^{\symMaxUpdateValMax - \symMaxUpdateValMin}
			\symLikelihoodFuncExponentTwo_{\symMaxUpdateValMax - \symIndexJ} \symOtherFunc(\symXZero, 2^\symIndexJ)}
		{
			\sum_{\symIndexJ=0}^{\symMaxUpdateValMax - \symMaxUpdateValMin}
			\symLikelihoodFuncExponentTwo_{\symMaxUpdateValMax - \symIndexJ}}
		=
		0.
	\end{multline*}
	Resolving for $\symXZero$ finally gives the lower bound. The upper bound results from \eqref{equ:inequ_proof} when using $\symOtherFunc(\symX, \symY)\leq 1$ that is a consequence of Bernoulli's inequality.
\end{proof}
\bibliographystyle{ACM-Reference-Format}
\bibliography{bibliography}

\end{document}